\documentclass[12pt]{report}
\usepackage{suthesis-2e}
\dept{Mathematics \& Applied Mathematics}

\usepackage{epsfig}
\renewcommand{\subset}{\subseteq}

\newcommand{\e}{\epsilon}

\newcommand{\pse}{\psi^\epsilon}

\newcommand{\ue}{u^\epsilon}

\newcommand{\cpsi}{\overline{\psi}}
\newcommand{\we}{W^\e }
\newcommand{\fe}{f^{\e}}

\newcommand{\tw}{\widetilde{W}^{\e}}
\newcommand{\weo}{W^{0}}

\newcommand{\si}{\sigma}

\newcommand{\lo}{\Lambda_0}
\newcommand{\lt}{\Lambda_t}

\newtheorem{theorem}{Theorem}[section]
\newtheorem{lemma}[theorem]{Lemma}
\newtheorem{proposition}[theorem]{Proposition}

\newenvironment{proof}[1][Proof]{\begin{trivlist}
\item[\hskip \labelsep {\bfseries #1}]}{\end{trivlist}}
\newenvironment{definition}[1][Definition]{\begin{trivlist}
\item[\hskip \labelsep {\bfseries #1}]}{\end{trivlist}}

\newenvironment{remark}[1][Remark]{\begin{trivlist}
\item[\hskip \labelsep {\bfseries #1}]}{\end{trivlist}}

\newcommand{\qed}{\nobreak \ifvmode \relax \else
      \ifdim\lastskip<1.5em \hskip-\lastskip
      \hskip1.5em plus0em minus0.5em \fi \nobreak
      \vrule height0.75em width0.5em depth0.25em\fi}

\usepackage{color}

\begin{document}
\title{{\bf Asymptotic approximation} \\{\bf of the Wigner function}\\
  {\bf in two-phase geometric optics}}
\author{Konstantina-Stavroula Giannopoulou}
 \submitdate{November 2009}

\programthesis
                          \dept{``Mathematics and its applications"}

\principaladviser{George N. Makrakis}

\firstreader{Stathis Filippas}

\thirdreader{George Kossioris} %if needed
\copyrightfalse
 \beforepreface
\prefacesection{Abstract}
 We propose a  renormalization process of a two phase WKB
  solution, which is based on an appropriate surgery of local uniform asymptotic
  approximations of the Wigner transform of the WKB solution. We explain in details
how this process
   provides the correct spatial variation and frequency scales of the wave
  field on the caustics where WKB method fails. The analysis has been
  thoroughly presented in the case of a fundamental problem, that is the
  semiclassical Airy equation, which arises from the model
  problem of acoustic propagation  in a layer with linear variation of the sound speed.

%\input{abstract}
%have a separate abstract.tex file so you can also input it into
%the microfilm.tex template
%\prefacesection{Preface}
 % This thesis tells you all you need to know about...
\prefacesection{Acknowledgments}
Many thanks to my Teacher Mr G.Makrakis for everything he taught me,\\
for encouranging me, especially for giving me support in hard times.

\newpage
\begin{flushright}
To my parents,
\\
Maria and Dimitris
\end{flushright}
\afterpreface
%%%%%%%%%%%%%%%%%%%%%%%%%%%%%%%%%%%%%%%%%%%%%%%%%%%%%%%%%%%%%%%%%%%%%%%%%%%%%%%%%%%%%%%%%%%%%%%%%%%%%%%%%%%%%%%%%%%%%%%%%%%%

%%%%%%%%%%%%%%%%%%%%%%%%%%%%%%%%%%%%%%%%%%%%%%%%%%%%%%%%%%%%%%%%%%%%%%%%%%%%%%%%%%%%%%%%%%%%%%%%%%%%%%%%%%%%%%%%%%%%%%%%%%%%

\chapter{Introduction}

High-frequency wave propagation in inhomogeneous media, has been
traditionally investigated employing the method of geometrical
optics. Not only it is used to draw a qualitative picture of how
the energy propagates, but also to evaluate the wave fields
quantitatively. However, geometrical optics fails either on
caustics and  focal points where it predicts infinite wave
amplitudes, or in shadow regions (i.e. regions devoid of rays)
where it yields zero fields. On the other hand, formation of
caustics is a typical situation in optics, underwater acoustics
and seismology, as a result of multipath propagation from
localized sources. Indeed, even in the simplest oceanic models and
geophysical structures (see, e.g. Tolstoy and Clay \cite{TC},
Chapt. 5, and C\u{e}rven\`{y}, et.al. \cite{CMP}, Chapt. 3,
respectively) a number of caustics occur, depending upon the
position of the source and the stratification of the wave
velocities.

\section{Caustics and phase-space methods}

Mathematically, caustic surfaces are envelopes of rays. Physically, these surfaces are distinctive in that the field
intensity increases on them sharply as compared with the adjacent
regions. The rise of field is best of all seen at the focal points
where all the rays corresponding to the converging wave front
intersect.

In his classical book Stavroudis \cite {Sta} has remarked that in
contrast to rays and wave fronts, the caustic is one of the few
objects in optics that can be observed in reality. This remark
emphasizing the role of caustics, of course, has its own range of
validity, and it is true only to the point that in the close
vicinity of caustics, one can observe or measure a concentration
of the field. Is the caustic real in the above sense in all
situations, i.e., is the effect of field buildup on a caustic
appreciable enough for instruments to reveal, separate and
identify the caustics? This question may be answered with
heuristic criteria. A caustic may be deemed real, i.e., observed
or recorded, if the amplitude on the caustic is at least a few
times the mean field value elsewhere and the near-caustic zones of
adjacent caustics do not completely overlap. Some other conditions
of practical character like, noise should not be high, resolution
and sensitivity of the instruments should be sufficient, should be
satisfied.

Moving across a caustic gives  birth or annihilation of a pair of
rays at a time, and this discontinuous variation of the number of
rays across a caustic is qualified as a {\it catastrophe}. This
new and fruitful approach to caustics, developed only in the
recent years, allows a universal classification of the typical
caustics (see, e.g., \cite{KO1}).

From specific examples allowing exact solutions, it has been known
that the phase of the wave fields change by $-\pi/2$ upon touching
a smooth (nonsingular) caustic, and by $-\pi$ after passing a
three-dimensional focus. However, a universal rule on the
additional phase shift at a caustic has been formulated only in
the comparative recent works of Maslov (\cite{Ma1}, \cite{Ma2})
and Lewis \cite{Le},  although the germ of the idea goes back to
Keller \cite{Kel}. The formulation is based on the
stationary-phase approximation of certain diffraction integrals,
and it finally leads to the notion of the so-called {\it Maslov
trajectory index}, in the general case of multiple caustic
reflections.

Because the wave amplitude predicted by geometrical optics is
infinite on the caustics, as a result of ray convergence,
geometrical optics is inapplicable within a close neighborhood of
the caustic, as actual wave fields are always finite. However,
available exact and approximate solutions for certain canonical
wave problems involving caustics in the high-frequency limit,
indicate a substantial concentration of energy near a caustic.
This phenomenon is more profound within a finite region which is
usually referred as caustic zone or caustic volume. The rigorous
estimation of the size of this zone should rely upon delicate
uniform asymptotic expansions of certain canonical diffraction
integrals associated with the particular caustic, but for the
moment only heuristic estimations leading rather to qualitative
than to fully quantitative results exist. A very important feature
is that the rays cannot be adequately resolved in the caustic
zone, and therefore we can draw the general conclusion that within
any caustic zone, no physical device is capable of separate
determination of ray parameters. In this sense, in that caustic
zone, rays loose their physical individual properties, though they
continue to play the role of the geometric framework for the wave
field.

From the mathematical point of view, formation of caustics and the
related multivaluedness of the phase function, is the main
obstacle in constructing global high-frequency solutions of the
wave equation. The problem of obtaining the multivalued phase
function is traditionally handled by resolving numerically the
characteristic field related to the eikonal equation (ray tracing
methods), see, e.g. \cite{CMP}. A considerable amount of work has
been done recently on constructing the multivalued phase function
by properly partitioning the propagation domain and using eikonal
solvers (see, e.g., \cite{Ben}, \cite{FEO}).

Given the geometry of the multivalued phase function, a number of
local and uniform methods to describe wave fields near caustics
have been proposed. The local methods are  essentially based on
{\it boundary layer} techniques as they were developed by Babich,
Keller, et.al. (see, e.g., \cite{BaKi}, \cite{BB}). The uniform
are those which exploit the fact that even if the family of rays
has caustics, there are no such singularities for the family of
the bicharacteristics in the phase space. This basic fact allows
the construction of formal asymptotic solutions (FAS) which are
valid also near and on the caustics. For this purpose two main
asymptotic techniques have been developed. The first one is the
{\it Kravtsov-Ludwig method} (sometimes called the {\it method of
relevant functions}). This method starts with a modified FAS
involving Airy-type integrals, the phase functions of which take
account of the particular type of caustics. The second one is the
{\it method of the canonical operator} developed by Maslov. The
construction of the canonical operator
 exploits the fact that the Hamiltonian flow associated with
the bichararcteristics generates a Lagrangian submanifold in the
phase space, on which we can ``lift'' the phase function in a
unique way (see, e.g., \cite{MF},\cite{MSS},\cite{Va1}).

Although uniform caustic asymptotics have been widely used by the
acoustical and seismological community (see, e.g., \cite{CH1},
\cite{CH2} and the references cited there), the problem of the
limits of applicability of uniform asymptotic expressions has not
been completely resolved yet, as it has been observed by Asatryan
and Kravtsov \cite{AsK} who attempted to give a qualitative
answer. Note also that apart from their importance in the
qualitative description of wavefields near caustics, the
Kravtsov-Ludwig solutions have been also proved useful for
numerical computations through appropriate matching with geometric
optics far away from the caustics \cite{KKM}.

\section{ The Wigner-function approach}
 A relatively new technique to treat high frequency dispersive problems is based on
the Wigner transform of the wavefunction,
 whose basic properties (i.e. the relation of its moments with important physical
quantities, as energy density,
 current density, et.al.), make it a proper and extremely useful tool for the study
of the wavefield.
 Wigner function is a phase space object satisfying an integro-differential equation
(Wigner equation), which for smooth
 medium properties  can be expressed as an infinite order singular perturbation
(with dispersion terms with respect
 to the momentum of the phase space)of the classical
  Liouville equation. At the high frequency limit the solution of the Wigner
equation equation converges weakly to the so called
 Wigner measure \cite{LP} governed by the classical Liouville equation, and this
measure, in general, reproduces the solution of
 single phase geometrical optics.

We should note at this point that there does not exist, up to now,
either some systematic theoretical study of the Wigner
integro-differential equation (except the results of Markowich
\cite{SMM} for the equivalence of Wigner and   Schr\"odinger
equations). This is due to fundamental difficulties of this
equation, which is an equation with non-constant coefficients,
that combines at least two different characters, that of transport
and that of dispersive equations. The first character is
correlated with the Hamiltonian system of the Liouville equation
(and the classical mechanics of the problem), and the second with
the quantum energy transfer away from the Lagrangian manifold of
the Hamiltonian system, but mainly inside a boundary layer around
it, the width of which depends on the smoothness of the manifold
and the presence or not of caustics.

Moreover, in the case of multi-phase optics and caustic formation,
Wigner measure is  not the appropriate tool for the study of the
semi-classical limit. In fact ut has been shown by Filippas \&
Makrakis   \cite{FM1}, \cite{FM2} through examples {\it in the
case of time-dependent Schrodinger equation} that the Wigner
measure (a) it cannot be expressed as a distribution with respect
to the momentum for a fixed space-time point,  and thus cannot
produce the amplitude of the wavefunction, and (b) it is unable to
``recognize'' the correct frequency dependencies of the wavefield
near caustics. It was however explained  that the solutions of the
integro-differential Wigner equation do have the capability to
capture the correct frequency scales. It must be said here that a
numerical approach based on classical Liouville equation has been
developed, as an alternative to WKB method, in order to capture
the multivalued solutions far from the caustic. This technique is
based on a {\it closure assumption} for a system of equations for
the moments of the Wigner measure (essentially by assuming a fixed
number of rays passing through a particular point) (see, e.g.,
(\cite{JL}, \cite {Ru1}, \cite {Ru2}).

\section{The present work: Renormalization of WKB solutions}

In the present work we employ Wigner transform as a tool for the
renormalization of WKB solutions near caustics (wignerization). More precisely, we
consider the fundamental example of the semiclassical Airy
equation, whose two-phase WKB solution fails at the caustic (namely
the turning point of the Airy equation) due to
the divergence of the geometric amplitudes. We show
that the  combination (``surgery") of appropriate asymptotic approximations of the
Wigner function in various areas of the phase space  leads to an approximate Wigner
function
which recovers the correct semiclassical (Airy) amplitude in a spatially uniform
way. Moreover,
the interaction mechanism between the
two geometric phases (realized as the two branches of the folded
Lagrangian manifold) is investigated by  thoroughly analyzing the
structure of the stationary points of the corresponding cross
Wigner integrals and their asymptotic contribution in the Airy structure.
Note that for our particular examples, it happens   that the asymptotics of the
Wigner function leads to the exact Wigner transform of the
semiclassical Airy function, which confirms the validity of the proposed wignerization.
It should be emphasized that the proposed renormalization process has many
similarities and it has been inspired by the so-called
 quantization processes (see, e.g., Nazaikinskii et al \cite{NSS}), since we can
consider the WKB solution
as the ``classical object" and the constructed approximation of the Wigner function
as the ``quantum object". It is then interesting to
observe that in the full high-frequency limit the approximation of the Wigner
function, the quantum object, gives us back the WKB solution, that is the classical
object.

The structure of the work is the following. In Chapter 2 we present the technique of
geometric optics(WKB solutions) and the
method developed by Kravtsov \& Ludwig. We also analyze the propagation of plane
acoustic waves in a layer with
linear variation of the reafraction index, we construct the WKB and Kravtsov-Ludwig
solutions and we explain how the semiclassical Airy eqaution
bounces out of this model. In Chapter 3 we introduce the Wigner transform, we review
its basic properties and we also present in details
Berry's construction of the semiclassical Wigner function. In Chapter 4, which is
our main contribution, we construct the asymptotics
of the Wigner transform of the WKB solution of the semiclassical Airy equation and
we make some remarks on the stationary Wigner function
and the possibility of generalizing our process in problems with more complicated
refraction indices as well in two-dimensional propagation
problems where folds can also appear. Finally, in a series of four appendices we
briefly append basic results for the uniform stationary
phase method and the
idea of the parametrization of the Lagrangian manifold used by Kravtsov and Ludwing
in the derivation of their asymptotic
formula.

%%%%%%%%%%%%%%%%%%%%%%%%%%%%%%%%%%%%%%%%%%%%%%%%%%%%%%%%%%%%%%%%%%%%%%%%%%%%%%%%%%%%%%%%%%%%%%%%%%%%%%%%%%%%%%%%%%%%%%%%%%%%
%%%%%%%%%%%%%%%%%%%%%%%%%%%%%%%%%%%%%%%%%%%%%%%%%%%%%%%%%%%%%%%%%%%%%%%%%%%%%%%%%%%%%%%%%%%%%%%%%%%%%%%%%%%%%%%%%%%%%%%%%%%%
 \chapter{Geometrical optics}
 \section{The WKB method}
We consider the propagation of $n$-dimensional time-harmonic
acoustic waves in a medium with variable refraction index
$\eta(\textbf{x})=c_0/c(\textbf{x})$, $c_0$ being the reference
sound velocity and $c(\textbf{x})$ the wave velocity at the point
${\bf{x}}=(x_1,...,x_n) \in M$, where $M$ is some unbounded domain
in $R^n_{\bf{x}}$. We assume that $\eta \in
C^{\infty}({R^n_{\bf{x}}})$ and $\eta>0$. The wave field
$u(\bf{x}, \kappa)$ is governed by the {\it Helmholtz} equation

\begin{equation}\label{helm}
\Delta u(\textbf{x},\kappa) +\kappa^2{\eta}^2({\textbf{x}})u({\bf{x}},\kappa)=f(\textbf{x})\ ,\quad {\bf{x}} \in M \ ,
\end{equation}
where $\kappa =\omega/c_0$ is the wavenumber ($\omega$ being the
frequency of the waves) and $f$ is a compactly supported source
generating the waves. We are interested in the asymptotic behavior
of $u(\bf{x},\kappa)$ as $\kappa\to\infty$ (i.e. for very large
frequencies $\omega$) , assuming that $\bf{x}$ remains in a compact
subset of $M$. Note that the asymptotic decomposition of
scattering solutions when simultaneously $|\textbf{x}|$ and
$\kappa$ go to infinity is a rather complicated problem, as, in
general, the caustics go off to infinity. This problem has been
rigorously studied in Vainberg \cite{Va}, when $M$ is a full
neighborhood of infinity and $\eta=1$ for $ |\textbf{x}|>r_0$ , 
$r_0$ being a fixed positive constant, and by Kucherenko \cite{Ku}
for the case of a point source (i.e., $f(\bf{x})=\delta(\bf{x})$) ,
under certain conditions of decay for $\eta(\bf{x})$ as $
|\bf{x}|\to \infty$ .

For fixed $\kappa>0$ there is, in general, an infinite set of
solutions of (2.1), and we thus need a radiation condition to guarantee
uniqueness (cf. \cite {CK} for scattering by compact
inhomogeneities, and \cite{Wed} for scattering by stratified
media). This condition is essentially equivalent to the assumption
that there is no energy flow from infinity, which in geometrical
optics is translated to the requirement that the rays must go off
to infinity (cf. \cite{PV}, \cite{Ca}).

 \begin{definition}
We say that
\begin{equation} \label{asexp}
u_N({\bf{x}},\kappa)=e^{i\kappa
S(\bf{x})}\sum_{\ell=0}^{N}(i\kappa)^{-\ell}A_{\ell}(\bf{x}) \ ,
\end{equation}
 where the phase $S$ and the amplitudes $A_{\ell}$
are real-valued functions in $C^{\infty}(R^n_{\bf{x}})$ , is a
{\it formal asymptoptic solution (FAS)} of $(\ref{helm})$, if it
satisfies the asymptotic equation
\begin{equation} \label{approxhelm}
(\Delta + \kappa^2
\eta^2({\bf{x}}))u_N({\bf{x}},\kappa)=O(\kappa^{-N_1}) \ ,\quad
\kappa\to\infty \ ,
\end{equation}
with $N_1 \rightarrow \infty$ as $N\rightarrow \infty$ , in a
bounded domain $|\textbf{x}|\le a$ , $a$ being a positive constant.
\end{definition}
According to the WKB procedure, we seek a FAS of $(\ref{helm})$ in
the form $(\ref{asexp})$. Substituting $(\ref{asexp})$ into
$(\ref{approxhelm})$, and separating the powers
$(i\kappa)^{-\ell}, \thinspace \ell=0,1,\dots$ , we obtain the {\it
eikonal } equation
\begin{equation} \label{eikonal} \left(\nabla
S(\bf{x})\right)^2=\eta^2(\bf{x}) \ ,
\end{equation}
for the phase function, and the following hierarchy of {\it
transport} equations

\begin{equation}
\label{transport} 2\nabla S({\bf{x}})\cdot\nabla A_0({\bf{x}}) +
\Delta S({\bf{x}})A_0({\bf{x}}) = 0\ ,
\end{equation}

\begin{equation}\label{htransport}
 2\nabla S({\bf{x}})\cdot\nabla A_{\ell}({\bf{x}}) + \Delta
S({\bf{x}})A_{\ell}({\bf{x}})= -\Delta A_{\ell-1}({\bf{x}}),\
\ell=1,2,\dots  \ ,
\end{equation}
for the principal and higher-order amplitudes, $A_0$ and $A_{\ell}$ , respectively.

\section{Rays and Caustics}

A standard way for solving the eikonal equation $(\ref{eikonal})$
is based on the use of bicharacteristics (see, e.g.,  \cite{Ho}, Vol.
I, Chap. VIII, and \cite{Jo}, Chap. 2). Let $H(\bf{x},\bf{k})$ be the {\it Hamiltonian}
function

\begin{equation}\label{hamiltonian}
H(\textbf{x},\textbf{k})=\frac12\left(\textbf{k}^2-\eta^2(\textbf{x})\right),
\quad \textbf{x}\in M, \quad \textbf{k}\in R^n \ ,
\end{equation}
corresponding to the Helmholtz equation $(\ref{helm})$, where
$\textbf{k}=(k_1,...,k_n)$ is the momentum, conjugate to the
position $\textbf{x}=(x_1,...,x_n)$ .

The associated Hamiltonian system reads as follows
\begin{eqnarray}\label{hamsys}
 \frac{d\bf{x}}{dt}&=&\nabla_{\bf{k}}
H(\bf{x},\textbf{k})=\textbf{k} \ , \nonumber \\
\frac{d\textbf{k}}{dt}&=&-\nabla_{\bf{x}}
H(\bf{x},\textbf{k})=\eta(\bf{x})\cdot \nabla_{\bf{x}}
\eta(\bf{x}) \ .
\end{eqnarray}
Here, since we deal with a time-independent problem, $t$ is simply
a time-like parameter which parametrizes the trajectories. For
$\textbf{k}=\nabla_{\bf{x}}S(\bf{x})$ , we see that
$H(\textbf{x},\textbf{k})=0$ gives just the eikonal equation
$(\ref{eikonal})$.

Let now $\lo$ be a manifold of dimension $n-1$ in $R^n$ ,
\begin{eqnarray*}
\lo =
\{\textbf{x}=\textbf{x}_{0}(\theta)\ ,\ \theta=(\theta_{1},...,
\theta_{n-1})\in U_{0}\subset {R}^{n-1}\} \ .
\end{eqnarray*}
For $t=0$ we specify on $\lo$ the following
initial conditions
\begin{equation}\label{initbich}
 \textbf{x}(0)=\textbf{x}_{0}(\theta)\ ,\quad
\textbf{k}(0)=\textbf{k}_{0}(\theta)\ ,\quad \theta\in U_0
 \ ,
\end{equation}
for the Hamiltonian system $(\ref{hamsys})$, and for this reason, in the sequel, we
refer to $\lo$  as the initial manifold (from which the bicharacteristics depart).

Also, we assume the initial conditions
\begin{equation}
 S(\textbf{x})=S_{0}(\theta) \ , \quad
A_l(\textbf{x})=\alpha_{l}(\theta)\quad \mathrm{for}
 \quad \textbf{x}=\textbf{x}_{0}(\theta)\in \lo \ ,
\end{equation}
for the phase and the amplitudes, $S_{0}(\theta)$ ,
$\alpha_{l}(\theta)$ being given smooth functions on the initial
manifold $\lo$ . Note that $\textbf{k}_{0}(\theta)$ must satisfy
the condition
 \begin{equation}
 {\mid \textbf{k}_{0}(\theta)\mid }^2=\Bigl(\eta
(\textbf{x}_{0}(\theta))\Bigr)^2, \quad \textbf{x}_0 \in \lo \ ,
\end{equation}
for the eikonal to be satisfied also at $t=0$ .

\begin{definition}
 The trajectories $\left\{\textbf{x}=\textbf{x}(t,\theta)\ , \
\textbf{k}=\textbf{k}(t,\theta)\ ,\ t\in R\ ,\ \theta\in U_0
\right\}$ which solve the initial value problem $(\ref{hamsys})$,
$(\ref{initbich})$,  in the phase space $R_{\bf{x} \bf{k}}^{2n}$
are called {\bf bicharacteristics}, and their projection
$\{\textbf{x}=\textbf{x}(t,\theta)\ ,\ t\in R\ ,\ \theta\in
U_0 \}$ onto $R_{\bf{x}}^n$ are called {\bf rays}.
\end{definition}

The initial manifold $\lo$ evolves under the Hamiltonian flow defined by the
bicharecteristics to the manifold
\begin{eqnarray*}
\lt= \Bigl\{\bigl(\textbf{x}(t,\theta),\textbf{k}(t,\theta)\bigr)
\ , \ \theta \in U_0 \ ,\ t\ge 0 \Bigr\} \ .
\end{eqnarray*}
Obviously, since
 $\textbf{k}=\nabla_{\textbf{x}} S$ , the bicharacteristics lie
on the manifold $H({\textbf{x}},{\textbf{k}})= 0$ , thanks to the
eikonal equation, and therefore $\lt$ is a subset of the
constant-energy manifold $H({\bf{x}},{\textbf{k}})=0$ for any $t
\ge0$ . Moreover, in order to the eikonal to be satisfied for
$t=0$ , it must be $\textbf{k}_0=\nabla_{{\bf{x}}} S_0 (\textbf{x})$ and
therefore $\lo$ has the important property that it is a Lagrangian
manifold in the phase space $R_{\bf{x} \bf{k}}^{2n}$ (see, the
book by Maslov \& Fedoryuk \cite{MF} for a detailed introduction
to the theory of Lagrangian manifolds and its relation to the
construction of asymptotics, and also the expository paper by
Littlejohn \cite{Lit}). This property remains invariant under the
Hamiltonian flow, and therefore $\lt$ remains Lagrangian for all
$t \ge 0$ .

Also, in the sequel, we will assume that that $\textbf{k}_{0}(\theta)$ is nowhere
tangent to
$\lo$  in order  to the (non-characteristic) Cauchy problem for
the eikonal equation  have a smooth unique solution for small $t$ (see, e.g. Courant
\& Hilbert \cite{CH}, Vol. II, Chap. II ).

Now, using the first equation of the Hamiltonian system
$(\ref{hamsys})$
\begin{equation}
\frac{d\textbf{x}}{dt}= \textbf{k}= \nabla S_{{\bf{x}}} \\
\end{equation}
we see that $S$ satisfies the following ordinary differential
equation
\begin{equation}
\frac{dS(\textbf{x})}{dt}=\nabla_{{\bf{x}}} S \cdot
\frac{d\textbf{x}}{dt}=
\textbf{k}\frac{d\textbf{x}}{dt}=|\textbf{k}|^{2}=
\eta^2(\textbf{x}) \ .
\end{equation}
Integrating the last equation along the rays, we obtain the phase
\begin{equation}\label{gphase}
S(\textbf{x}(t,\theta))=S_{0}(\theta )
 +\int_0^t \eta^2 (\textbf{x}(\tau,\theta))d\tau \ .
\end{equation}

The solution of the  transport equation $(\ref{transport})$ for
the principal amplitude $A_0$ along the rays, is obtained by
applying divergence theorem in a ray tube $T_t$ . Assuming that
$A_0$ is finite and non-zero, the transport equation
$(\ref{transport})$ is rewritten in the divergence form
\begin{equation}
\nabla\cdot ({A_0}^2 \nabla S)=0 \ ,
\end{equation} and integrating on the ray tube $T_t$
with boundary $\partial T_t=\Sigma_0 \cup \Sigma_{0t} \cup
\Sigma_t$ , we have
\begin{equation}
 0=\int_{T_t}\nabla \cdot ({A_0}^2\nabla S)\, dT=\int_{\partial
T_t}^{}{A_0}^2(\nabla S\cdot \vec{\nu})\, d\Sigma \nonumber
 \end{equation}
where $\vec{\nu}$ the unit outer normal vector on the boundary
$\partial {T_t}$ of $T_t$ , and $d\Sigma$ is surface element.

Now, since the rays have the direction of $\nabla S$ , that is, the
rays are perpendicular to the wave fronts $S=$ const. , the vectors
$\vec{\nu}$ and $\nabla S$ are orthogonal on the lateral boundary
$\Sigma_{0t}$ of the ray tube, and we therefore obtain

\begin{eqnarray}
0=\int_{\partial T_t}^{}{A_0}^2(\nabla S\cdot \vec \nu)\, 
d\Sigma&=&\int_{\Sigma_{0}}\alpha_0^2 (\nabla S\cdot \vec \nu)\,
d\Sigma_{0} +\int_{\Sigma_t}{A_0}^2(\nabla S
\cdot \vec \nu)\, d\Sigma_t \nonumber \\
&=&-\int_{\Sigma_{0}}\alpha_0^2 |\textbf{k}_{0}|\, d\Sigma_{0}
+\int_{\Sigma_t}{A_0}^2 |\textbf{k}|\, d\Sigma_t \ .
\end{eqnarray}

Then, we have
\begin{equation}
 -\alpha_{0}^2|\textbf{k}_{0}|\, d\Sigma_0
+{A_0}^2|\textbf{k}|\, d\Sigma_t=0 \nonumber
\end{equation}
that gives
 \begin{equation}
 {A_0}^2
=\alpha_{0}^2\frac{|\textbf{k}_0|}{|\textbf{k}|}\frac{d\Sigma_0}{d\Sigma_t}=
\frac{\alpha_{0}^2}{J}
\end{equation}
where
\begin{equation}\label{jacobian}
 J=J(t,\theta )=\frac{D(t,\theta)}{D(0,\theta)} \ , \qquad D(t,\theta
)=\det \frac{\partial \textbf{x}(t,\theta)}{\partial(t,\theta )}\
,
\end{equation}
is the Jacobian of the ray transformation $(t,\theta)\mapsto
\textbf{x}(t, \theta)$ (see, e.g., \cite{BB}, \cite{Zau}).

Therefore we derive the principal amplitude
\begin{equation} \label{amplitude}
A_0(\textbf{x}(t,\theta))=\frac{\alpha_0(\theta)}{\sqrt{J(t,\theta)}}
\end{equation}
where $A_0(\textbf{x}(t=0,\theta))=\alpha_0(\theta)$ is the
amplitude at the point $\textbf{x}=\textbf{x}_0(\theta)$ on the
initial manifold $\lo$ .

\noindent
\begin{remark}
An alternative way to derive the formula $(\ref{amplitude})$
starts form the Liouville formula \cite{Ha}

\begin{equation} \label{smirn}
 \frac{d}{dt}\ln\frac{D(t, \theta)}{D(0,
\theta)}=\sum_{i=1}^{n}{\frac{\partial k_i}{\partial
x_i}}=\sum_{i=1}^{n}{\frac{\partial^2 S}{\partial x_i^2}}=\Delta S
\end{equation}
which, since $D(0, \theta)=1$ , implies
\begin{equation}
 \frac{d}{dt}\ln{D(t, \theta)}=\Delta S \ .
\end{equation}

From the  transport equation $(\ref{transport})$ we have
\begin{eqnarray*}
-\Delta S=\frac{2}{A_0}\nabla S\cdot \nabla A_0
\end{eqnarray*}
and since
\begin{eqnarray*}
\frac{d}{dt}\ln
{A_0}^2=\frac{2}{A_0}\frac{dA_0}{dt}=\frac{2}{A_0}\nabla
A_0\frac{d\textbf{x}}{dt}=\frac{2}{A_0}\nabla A_0\cdot \nabla S
\end{eqnarray*}
it follows
\begin{eqnarray*}
\frac{d}{dt}\ln {A_0}^2=-\Delta S \ .
\end{eqnarray*}
Then using $(\ref{smirn})$ we get
\begin{eqnarray*}
\frac{d}{dt}\ln {A_0}^2=-\frac{d}{dt}\ln D(t) \ ,
\end{eqnarray*}
which after integration on the interval $(0,t)$ leads to the
formula $(\ref{amplitude})$ for the principal amplitude.
\end{remark}

The higher-order amplitudes can be also derived by integrating the hierarchy of
transport equations $(\ref{htransport})$ in a similar way.

The transformation
\begin{eqnarray*}
(t,\theta)\mapsto \textbf{x}(t, \theta) \ ,
\end{eqnarray*}
is one-to-one, provided that the determinant of the Jacobian
\begin{eqnarray*}
D(t,\theta)=\det\frac{\partial(x_1, ..., x_n)}{\partial(t,
\theta_1, ..., \theta_{n-1})} \ ,
\end{eqnarray*}
is non-zero. Note that $D(t,\theta)$ is non-zero since we have excluded the
possibility of the characteristics to be tangent to $\lo$ at $t=0$ . But even if
$D(t=0,\theta) \ne 0$ , $D(t,\theta)$ and therefore $J(t,\theta)$ ,  it does
not necessarily remain non-zero for all $t$ . Whenever $J=0$ , it
can happen that $(t,\theta)$ may be non-smooth or multi-valued
functions of $\textbf{x}$, and the rays may intersect, touch, and
in general have singularities, although the bicharacteristics
never develop such singularities in the phase space. Then, the phase function
$S=S(\textbf{x}(t,\theta))$ may be a multi-valued or even a
non-smooth function. It must be emphasized that in the
neighborhoods of the singular points we cannot choose the
coordinates $x_1,\ldots,x_n$ as local coordinates.

\begin{definition}
The points $\textbf{x}=\textbf{x}(t,\theta)$ at which
$J(t,\theta)=0$ , are called {\bf focal points}, and the manifold
generated from these points, that is, the envelope of the family
of the rays, is called {\bf caustic}.
\end{definition}

It follows from ($\ref{amplitude}$), that the principal amplitude
$A_0$  blows up on the caustics. However, it is known that
solutions of the Helmholtz equation are analytic away form source
points and it is therefore the WKB procedure for constructing the
(FAS) which fails to predict the correct amplitudes on the
caustics. From the geometrical point of view, this non-physical
blow up of the amplitude at the caustic is associated with the
diminishing of the ray tubes there (the ray tube cross section
$\Sigma_t$ vanishes whenever the ray touches the caustic), and it
is clearly a consequence of the way of solving the  transport
equation by integrating along the rays. In fact, a boundary layer
analysis \cite {BaKi}, \cite {BuKe} shows that the ray structure
breaks down near the caustic, and within a boundary layer the
modal structure of the wave field is dominant, which makes
therefore impossible to separate the waves approaching the caustic
from those leaving from it. However, uniform asymptotic solutions
which will be considered in the sequel, show that it exists
considerable energy concentration near the caustic, which makes it
detectable, but the field is spatially finite but  strongly
increasing with increasing frequency.

Asymptotic methods for calculating finite fields on the caustics
have been developed by Kravtsov \cite{Kra}, \cite{KO} and Ludwig
\cite{Lu} (the method of relevant functions) and by Maslov
\cite{MF}(the method of the canonical operator). Although the two
methods have been developed along different lines, they are both
essentially based on the symplectic properties of the
Lagrangian manifold  $\lt$ . We will briefly present the
Kravtsov-Ludwig technique.
%%%%%%%%%%%%%%%%%%%%%%%%%%%%%%%%%%%%%%%%%%%%%%%%%%%%%%%%%%%%%%%%%%%%%%%%%%%%%%%%%%%%%%%%%%%%%%%%%%%%
A relatively recent way to treat high frequency problems is based
on the Wigner transform of the wavefunction, whose basic
properties (i.e. the relation of its moments with important
physical quantities, as energy density, current density, et.al.),
make it a proper and extremely useful tool for the study of the
wavefield. Wigner function is satisfying an integro-differential
equation in phase space, which for smooth potential functions can
be expressed as an infinite order singular perturbation (with
dispersion terms with respect to the momentum of the phase space),
of the classical Liouville equation.

 At the high frequency limit, the solution of Liouville equation
converges weakly to the so
 called Wigner measure \cite{LP}, which for relatively smooth initial
phase functions $S_0$ produces
 the solution of single phase geometrical optics. But in the case of
multi-phase optics and caustic formation,
 Wigner measure is not the appropriate tool for the study of the
semi-classical limit, because as is shown through
 examples for the time-dependent Schrodinger equation by Filippas \& Makrakis  
\cite{FM1}, \cite{FM2} (a) it cannot
be expressed as a distribution with respect
 to the momentum for a fixed space-time point, and thus cannot produce the
amplitude of the wavefunction,
 and (b) it is unable to ``recognize" the correct dependencies of the
wavefield from the semiclassical parameter
 $\e$ near caustics. This  is a property that the solutions of the
integro-differential Wigner equation do have.

We should note  that up to now, there does not exist either some
systematic theoretical study of the Wigner integro-differential
equation (except the results of Markowich \cite{SMM} for the
equivalence of Wigner and   Schr\"odinger equations), neither some
method for constructing solutions or their representations. This
 is due to fundamental difficulties of this equation,
which is an equation with non-constant coefficients, that combines
at least two different characters, that of transport  and that of
dispersive equations. The first character is correlated with the
Hamiltonian system of the Liouville equation (and the classical
mechanics of the problem), and the second with the wave energy
transfer away from the Lagrangian manifold of the Hamiltonian
system- mainly inside a boundary layer around the manifold- the
width of which depends on the smoothness of the manifold and the
presence or not of caustics.
%%%%%%%%%%%%%%%%%%%%%%%%%%%%%%%%%%%%%%%%%%%%%%%%%%%%%%%%%%%%%%%%%%%%%%%%%%%%%%%%%

\section{The Kravtsov-Ludwig technique}

\subsection{Motivation and heuristic foundation}

The idea of obtaining global high-frequency solutions of the
Helmholtz equation $(\ref{helm})$ by the method of relevant
functions, is to replace $(\ref{asexp})$ by integrals of the form
(see, e.g., the classical paper by Ludwig  \cite {Lu}, the modern approach by
Duistermaat \cite {Dui1},  \cite{Dui2}, also the book by Guillemin \& Sternberg \cite{GS})
\begin{equation}\label{intexp}
u(\textbf{x},
\kappa)=\left(\frac{i\kappa}{2\pi}\right)^{\frac12}\int_{\Xi}e^{i\kappa
S\left(\bf{x},\xi\right)}\, A(\textbf{x},\xi)\, d\xi \ , \quad \textbf{x}
\in M\subset R_{\bf{x}}^n\ , \quad \xi\in\Xi\subset R_{\xi}\ .
\end{equation}
where the phase $S(\textbf{x},\xi)$ and the amplitude
$A(\textbf{x},\xi)$ must satisfy  the eikonal equation
$(\ref{eikonal})$ and the transport equation $(\ref{transport})$,
respectively, identically with respect to $\xi$. Here, the term global solution
means that the integral representation holds uniformly away and on the caustics.

The integral $(\ref{intexp})$ can be regarded as a continuous
superposition of oscillatory functions of the form
$(\ref{asexp})$. The underlying physical motivation is the fact
that in every small region in which the refraction index of the
medium can be approximately considered as constant, and the wave front as plane,
the field can be represented as a superposition of plane waves
$Ae^{i\kappa S}$ , where the local amplitude $A$ and the local wavenumber
$\nabla_{\bf{x}} S$ vary slowly in
transition from one region to the next.

The phase $S(\textbf{x},\xi)$ {\it parametrizes the Lagrangian
submanifold} $\lt$ generated by the corresponding Hamiltonian
flow, in the sense that $\lt$ is locally represented by
$(\textbf{x}, \textbf{k})= (\textbf{x},
\nabla_{\bf{x}}S(\textbf{x},\xi))$ . In the language of microlocal analysis, the
representation
$(\ref{intexp})$ defines a Lagrangian distribution on $\Lambda$,
which for large $\kappa$ is an asymptotic solution (compound
asymptotics) of the Helmholtz equation (see the book by Guilllemin
and Sternberg \cite{GS} for a detailed but rather technical
exposition of this technique). In this sense, the construction of
an asymptotic expansion in the form $(\ref{intexp})$ is
``equivalent" with the construction of the Lagrangian submanifold
$\Lambda$. Near caustics $S(\textbf{x},\xi)$ is a multivalued
function and, in general, it cannot be derived by integration
along the bicharacteristics by simply applying $(\ref{gphase})$.
Representation formulae for the phase function $S(\textbf{x},\xi)$
are constructed, for each caustic which is generated from the
particular ray system (different caustics may appear for the same Hamiltonian with
different initial data), by appealing, in general, to the methods of
singularity theory (see, e.g., \cite {AVH}). For a simple fold caustic
 this construction is relatively simple, and we
briefly present it in the next section.

First of all, in the case of single-phase geometrical
optics, we can take $S(\textbf{x}, \xi)=\phi(\textbf{x})-\xi^2$ .
Then, $\partial_{\xi}S(\textbf{x}, \xi)=-2\xi$ , and there is only
a simple stationary point $\xi=0$. By the standard stationary
phase lemma (see, e.g., \cite{BH}, p. 219), the oscillatory
integral $(\ref{intexp})$ reduces asymptotically to
$(\ref{asexp})$. If there are more than one simple stationary
points $\xi_j(\textbf{x})$ , that is,
$\partial_{\xi}S(\textbf{x},\xi_j(\textbf{x}))=0$ and
$\partial_{\xi}^2S(\textbf{x},\xi_j(\textbf{x}))\neq 0$ , we obtain
the asymptotic expansion
\begin{equation}\label{wkbsum}
u({\bf{x}}, \kappa)\sim\sum_{j=1}^{j=J}  A_{0}^j({\bf{x}})e^{i\kappa
S_j({\bf{x}})}\ .
\end{equation}
Here
\begin{equation}
S_j(\textbf{x})=S(\textbf{x},\xi_j(\textbf{x})) \ ,
\end{equation}
solve the eikonal equation $(\ref{eikonal})$, and
\begin{equation}
A_{0}^j(\textbf{x})=e^{\frac{i\pi}{4} \bigl(1+sgn\partial_{\xi}^2
S({\bf{x}},\xi_j({\bf{x}}))\bigr)}
\frac{A(\textbf{x},\xi_j(\textbf{x}))}{\sqrt{|\partial_{\xi}^2S({\bf{x}},\
\xi_j(\bf{x}))|}}\ ,
\end{equation}
solve the zero-order transport equation $(\ref{transport})$.

Note that the existence of many stationary points $\xi_j(\textbf{x}) \ ,\ j= 1,
\cdots, J$ for some fixed point $\textbf{x}$, means that from this point pass $J$
rays, and $S_j(\textbf{x}) \ ,A_{0}^j(\textbf{x} )$ are the phase and the amplitude
computed by integrating the eikonal and the transport equations along the $j-$th ray
passing from $\bf{x}$. The summation in $(\ref{wkbsum})$ extends over all the rays,
a fact which implies that the principal asymptotic contribution to the  wavefield is
just the superposition of the individual geometric (WKB) wavefields, and there no
significant interference effects between these waves.
However, the above picture is not valid whenever
$\partial_{\xi}^2S(\textbf{x},\xi_j(\textbf{x}))=0$ , i.e. for the
stationary points of multiplicity greater than one. In this case,
a modified stationary phase lemma (\cite{BH},  p. 222,  \cite{Bor},  \cite{CFU}) must be
applied in order to obtain the  correct expansion. The appearance  of many 
stationary points which coalesce, is associated with the formation of caustics and
the interference effects between the local geometrical waves cannot be ignored, thus
making the modal structure of the wavefield important within a boundary layer
adjacent to the caustic.

\subsection{Phase functions for smooth caustic (folds)}

We  start by stating and briefly describing the proof of the
following basic proposition which can be found in the book by Guillemin \& Sternberg
(\cite{GS}, p.431, Proposition 6.1).

\begin{proposition}\label{reprfold}
Near a smooth caustic (fold), the phase function has the form
\end{proposition}
\begin{equation}\label{sxi}
S(\textbf{x},\xi)=\phi(\textbf{x})+\xi\rho(\textbf{x})-\frac{\xi^3}3
\ ,
\end{equation}
{\it{and the amplitude admits of the decomposition}}
\begin{equation}\label{axi}
A(\textbf{x},\xi)=g_0(\textbf{x})+\xi
g_1(\textbf{x})+h(\textbf{x},\xi)\left(\rho(\textbf{x})-
\xi^2\right) \ ,
\end{equation}
{\it where} $h(\textbf{x},\xi)$ {\it is a smooth function, and}
$\rho(\textbf{x})- \xi^2=\partial_{\xi}S(\textbf{x},\xi)$ .

Substituting $(\ref{sxi})$ and $(\ref{axi})$ into
$(\ref{intexp})$, integrating the first two terms, and estimating
by the standard stationary phase lemma the contribution of the third term,
we obtain the following uniform asymptotic expansion
\begin{equation}\label{klformula}
u(\textbf{x})=\sqrt{2\pi} \kappa^{\frac16} e^{\frac{i\pi}4}
e^{i\kappa\phi(\bf{x})} \Biggl(g_0(\textbf{x})
Ai\left(-\kappa^{\frac23}\rho(\textbf{x})\right)+i\kappa^{-\frac13}g_1(\textbf{x})
Ai^{\prime}\left(-\kappa^{\frac23}\rho(\textbf{x})\right)\Biggr)+O(\kappa^{-
1})\, ,
\end{equation}
$\kappa\to\infty  \ ,$ where $Ai(\cdot)$ is the {\it Airy
function}.

Now, by inserting the asymptotic expansions of the Airy functions
for large negative arguments (see, e.g., \cite{Leb}) into
$(\ref{klformula})$, we find for $\kappa\to\infty$ and $\rho \ne
0$ , the following geometrical-optics expansion of the field
\begin{equation}\label{asklformula}
\hspace{-0.15cm}
u(\textbf{x})=\frac1{\sqrt2}\Biggl(\left(g_0(\textbf{x})+g_1(\textbf{x})
\sqrt{\rho(\textbf{x})}\right)\rho^{-\frac14}e^{i\kappa
\Phi_{+}(\bf{x})}+
\left(g_0(\textbf{x})-g_1(\textbf{x})\sqrt{\rho(\textbf{x})}\right)\rho^{-\frac14}e^{i\kappa
\Phi_{-}({\bf{x}})+\frac{i\pi}2}\Biggr) 
\end{equation}
where
\begin{equation}\label{gklphases}
\Phi_{\pm}(\textbf{x})=\phi(\textbf{x})\pm\frac23\rho^{\frac32}(\textbf{x})\ .
\end{equation}

In order to define the Kravtsov-Ludwig amplitude and phases $\phi$, $\rho$, $g_0$ and $g_1$, 
we apply the so-called {\it asymptotic
matching principle}, which states that the expansion
($\ref{asklformula}$) must coincide with the WKB expansion
\begin{equation}\label{wkbfold}
u(\textbf{x})=A_+(\textbf{x})e^{i\kappa
S_+(\bf{x})}+A_-(\textbf{x})e^{i\kappa S_-(\bf{x})}\ ,
\end{equation}
away from the caustic and for large frequencies. This principle implies that
\begin{eqnarray}
\frac1{\sqrt2}\left(g_0(\textbf{x})+g_1(\textbf{x})\sqrt{\rho(\textbf{x})}\right)\rho^{-1/4}&=&
A_+(\textbf{x})\ , \\
\frac1{\sqrt2}\left(g_0(\textbf{x})-g_1(\textbf{x})\sqrt{\rho(\textbf{x})}\right)\rho^{-1/4}e^{\frac{i\pi}2}&=&
A_-(\textbf{x})\ ,
\end{eqnarray}
and
\begin{equation}
\Phi_{\pm}(\textbf{x})= S_{\pm}(\bf{x}) \ ,
\end{equation}
and therefore we obtain
\begin{eqnarray}\label{modamplitudes}
g_0(\textbf{x})&=&\frac{\rho^{\frac14}}{\sqrt2}\left(A_{+}(\textbf{x})-iA_{-
}(\textbf{x})\right)
\ ,\\
g_1(\textbf{x})&=&\frac{\rho^{-\frac14}}{\sqrt2}\left(A_{+}(\textbf{x})+iA_{
-}(\textbf{x})\right) \ ,
\end{eqnarray}
and
\begin{equation}\label{klphases}
\phi(\textbf{x})=\frac12\Bigl(S_{+}(\textbf{x}) +
S_{-}(\textbf{x})\Bigr) \quad \mathrm{and} \quad
\rho(\textbf{x})=\Biggl(\frac34\Bigl(S_{+}(\textbf{x}) -
S_{-}(\textbf{x})\Bigr)\Biggr)^{2/3} \ .
\end{equation}

Note that near the caustic, from any point M near the fold pass
two rays (see Figure 2.1). The subscript $(-)$ (respectively
$(+)$) indicates the ray which arrives at M directly from the
initial manifold (respectively, after ``reflection'' from the
caustic), and $A_{\pm}$ are the principal geometrical amplitudes
$A_0$ along the $(\pm)$ rays.

\begin{figure}[ht]
\centering
\includegraphics[width=0.4\textwidth]{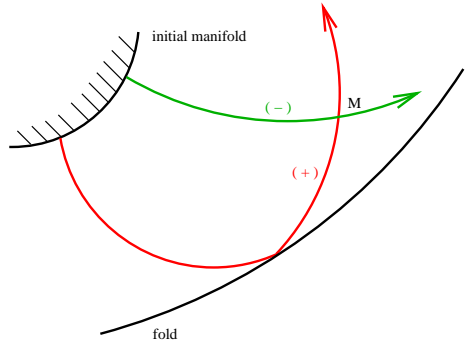}
\caption{{\it Initial manifold, rays \& caustic}}
\end{figure}

The geometrical amplitudes $A_{\pm}(\textbf{x})$ solve the
transport equations $(\ref{transport})$, and according to
$(\ref{amplitude})$ they are given by
\begin{equation}
A_{\pm}(\textbf{x})=\frac{\alpha_0(\theta_{\pm})}{\sqrt{J_{\pm}(\textbf{x})}
}\ ,
\end{equation}
where $\theta_{\pm}=\theta_{\pm}(\textbf{x})$ are the values of
the parameter at the initial manifold corresponding to the rays
$(\pm)$ passing from M, $\alpha_0(\theta_{\pm})$ are the
corresponding initial amplitudes, and $J_{\pm}(\textbf{x})$ are
the values of the Jacobian at the point $\textbf{x}$ calculated
along the $(\pm)$ rays. The value of the square root
$\sqrt{J_{\pm}}$ is given by the formula $\sqrt{J_{\pm}}=
\sqrt{\mid J_{\pm}\mid}e^{i {\frac{\pi}{2}} \gamma_{\pm}}$ where
$\gamma_{+}=1$ and $\gamma_{-}=0$ . Note that $\gamma_{\pm}$ is the
Maslov trajectory index, and it counts the number of tangencies of
 the rays with the caustic
along their course from the points $\textbf{x}_0(\theta_{\pm})$ on
the initial manifold to the point M. Moreover, the geometrical phases
$S_{\pm}(\textbf{x})$ can be
computed by  integration along the rays according to
$(\ref{gphase})$.

On the basis of the asymptotic formula $(\ref{klformula})$, Ludwig
\cite{Lu} has  drawn the following qualitative picture   of the wave field
near the fold:

\noindent i) At points in the illuminated zone whose distance from
the caustic is small compared with $\kappa^{-\frac23}$, the
predictions of geometrical optics are correct to order $-\frac12$ .

\noindent ii) The intensity of the field on the caustic is large
but finite (of order $\kappa^{\frac16}$).

\noindent iii) In the shadow zone there is an illuminated strip
(penumbra) of width of the order $\kappa^{\frac23}$ . It must be
emphasized here that WKB method fails to predict any penumbra as
the shadow zone is devoid of classical rays.

It is finally interesting to note that we can construct the equations satisfied by functions $\phi$, $\rho$, $g_0$ and $g_1$ entering the Kravtsov-Kudwig formula
$(\ref{klformula})$. For this, we substitute this formula into the Helmholtz
equation $(\ref{helm})$, and we ask for the equation to be
asymptotically valid for large $\kappa$. This procedure leads to
the following system for the Kravtsov-Ludwig phases
\begin{eqnarray}\label{klphasesyst}
&\left( \nabla_{\bf{x}} \phi \right)^{2} + \rho \left(
\nabla_{\bf{x}}
\rho \right)^{2} =\eta^{2}(x)  \ , \\
&\nabla_{\bf{x}} \phi \cdot \nabla_{\bf{x}}  \rho =0   \ .
\end{eqnarray}
The system for $g_0 $, $g_1$ is rather complicated and it is given
in \cite{Lu}, \cite{KO}.

\section{Example: Plane wave incident on linear layer}

We consider the two-dimensional {\it{Helmholtz}} equation
\begin{eqnarray*}
\Delta U (\textbf{x})+\kappa_0^2\eta
^2(z)U(\textbf{x})=0\ ,\quad\textbf{x}=(y,z)
\end{eqnarray*}
in a linearly stratified medium occupying the strip $0<z<h$, with
refraction index (see Figure 2.2)
\begin{eqnarray*}
\eta^2(z)=\mu_{0}+\mu_{1}z
\end{eqnarray*}
which increases with depth $z$ ($\mu_{1}>0$).

A plane wave of the form
\begin{eqnarray*}
U(y,h)=\exp(i\kappa_0 y\sin{\psi})
\end{eqnarray*}
arrives at the boundary $z=h$, with angle $\psi$ ($0<\psi<\pi/2$)
with respect to the vertical direction $z$.

\begin{figure}[ht]
\centering
\includegraphics[width=0.6\textwidth]{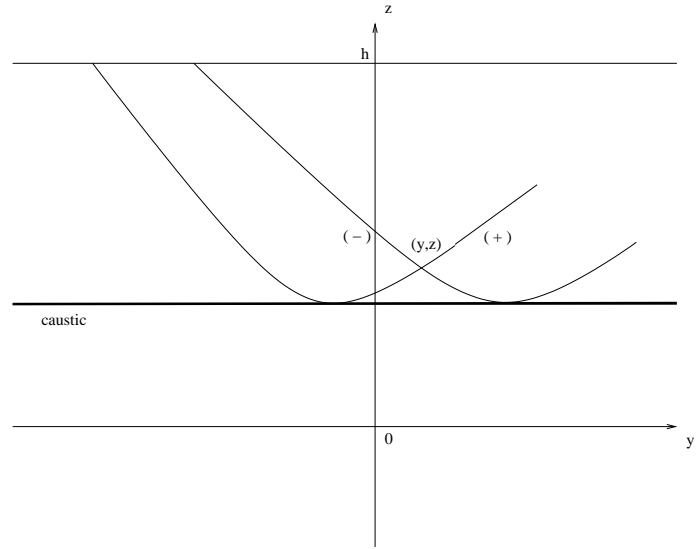}
\caption{{\it Caustic for a linear layer}}
\end{figure}

For this medium, the Hamiltonian function is

\begin{eqnarray*}
H(y,z,k_y,k_z)=\frac{1}{2}\Bigl(|\textbf{k}|^2-\eta^2(z)\Bigr),\quad
\textbf{k}=(k_y,k_z)
\end{eqnarray*}

and the equations of the rays are given by the Hamiltonian system
\begin{eqnarray}
\frac{dy}{dt}&=&k_y\ ,\quad y(0)=\xi \nonumber \\
 \frac{dz}{dt}&=&k_z\ ,\quad
z(0)=h  \nonumber\\
\frac{dk_{y}}{dt}&=&0\ ,\quad k_{y}(0)=\eta_0\sin{\psi} \nonumber \\
\frac{dk_z}{dt}&=&-\frac{\mu_1}{2}\ ,\quad k_z(0)=-\eta_0\cos\psi\ .
\end{eqnarray}
Solving the above system, we find the parametric equations of the
rays
\begin{eqnarray}\label{rays}
z&=&\frac{\mu_1}{4}t^2-\eta_{0}t\cos{\psi}+h \nonumber \\
y&=&\xi+\eta_{0}t\sin{\psi}\ .
\end{eqnarray}

The Jacobian of the ray transformation is given by
\begin{equation}
J=\frac{1}{\eta_{0}\cos{\psi}}\left| \matrix{%
\frac{\partial{y}}{\partial{t}}&\frac{\partial{z}}{\partial{t}}\cr
\frac{\partial{y}}{\partial{\xi}}&\frac{\partial{z}}{\partial{\xi}}
} \right|
=\frac{1}{\eta_0\cos{\psi}}\Bigl(-\frac{\mu_1}{2}t+\eta_0\cos{\psi}\Bigr)
\ ,
\end{equation}
and it vanishes on the caustic given by
\begin{equation}
z_c=h-\frac{1}{\mu_1}{\eta_0}^2\cos^2{\psi}\ .
\end{equation}

From the equations of the rays $(\ref{rays})$, and for any given
point $(y,z<h)$, we find two values of the parameter $t$, that is
\begin{eqnarray}\label{arrivals}
t_{+}(z)&=&\frac{2}{\mu_1}\Bigl(\eta_0\cos{\psi}+\sqrt{{\eta_0}^2\cos^2\psi
+\mu_1(z-h)}\, \Bigr)
\ , \\
t_{-}(z)&=&\frac{2}{\mu_1}\Bigl(\eta_0\cos{\psi}-\sqrt{{\eta_0}^2\cos^2\psi
+\mu_1(z-h)}\, \Bigr) \ ,
\end{eqnarray}
 and the corresponding initial positions
\begin{eqnarray}\label{inpos}
\xi_{+}(y,z)&=&y-\frac{2}{\mu_1}\eta_0\sin{\psi}\Bigl[\eta_0\cos{\psi}+\sqrt
{{\eta_0}^2\cos^2\psi+\mu_1(z-h)}\, \Bigr]
\ , \\
\xi_{-}(y,z)&=&
y-\frac{2}{\mu_1}\eta_0\sin{\psi}\Bigl[\eta_0\cos{\psi}-\sqrt{{\eta_0}^2\cos
^2\psi+\mu_1(z-h)}\, \Bigr] \ .
\end{eqnarray}

This means that from every given point $(y,z)$, at the times
$t_{\pm }$ pass two rays emanating from the points $\xi_{\pm }$ on
the illuminated boundary $z=h$. Note that for $z=z_c$ , we have
$t_{+}(z_c)=t_{-}(z_c)=:t_c$ .

Using now the equation $(\ref{gphase})$ we compute the phase
function
\begin{eqnarray*}
S(t)&=&\int_{0}^{t}\eta^2(z(\tau))d\tau+S(y(0),z(0))\\
&=&\int_{0}^{t}(\mu_0-\mu_1z(\tau))d\tau+S(\xi,h) \\
&=&\frac{{\mu_1}^2}{12}t^3-\frac{\mu_1\eta_0\cos\psi}{2}t^2+(\mu_0+\mu_{1}h)
t+\eta_0 \xi\sin{\psi}\ ,
\end{eqnarray*}
and substituting the values of  $t_{+}$ , $t_{-}$ , $\xi_{+}$ ,
$\xi_{-}$ , we obtain the geometrical phases
\begin{eqnarray}
S_{+}(y,z)&=&-\frac{6(\mu_0+\mu_{1}h)(\alpha-\beta)+6\alpha^2\beta-4\alpha^3-
2\beta^3}{3\mu_1}+\eta_0\xi_{+}\sin{\psi}
\ ,\\
S_{-}(y,z)&=&-\frac{6(\mu_0+\mu_{1}h)(\alpha+\beta)-6\alpha^2\beta-4\alpha^3+
2\beta^3}{3\mu_1}+\eta_0\xi_{-}\sin{\psi}
\end{eqnarray}
where
\begin{equation}
\alpha=-\eta_0\cos{\psi}\ , \quad
\beta=\sqrt{{\eta_0}^2\cos^2{\psi}+\mu_1(z-h)} \ .
\end{equation}

Using the equations $(\ref{klphases})$, we compute the
Kravtsov-Ludwig coordinates (modified phases)

\begin{equation}
\phi(y)=\frac{2}{3\mu_1}\eta_{0}^3\cos^3{\psi}+\eta_{0}y\sin{\psi}
\end{equation}
and
\begin{equation}
\rho(z)={\Bigl(\frac{1}{\mu_1}\Bigr)}^{2/3}(\alpha^2+\mu_1(z-h))\ .
\end{equation}

The Jacobians $J_{\pm}$ along the two rays passing from the point
$(y,z)$, are given by
\begin{eqnarray}
J_{+}&=&-\frac{1}{\eta_{0}\cos{\psi}}\sqrt{{\eta_0}^2\cos^2\psi+\mu_1(z-h)}\ ,
\\
J_{-}&=&\frac{1}{\eta_{0}\cos{\psi}}\sqrt{{\eta_0}^2\cos^2\psi+\mu_1(z-h)}
\end{eqnarray}
and the corresponding principal geometrical amplitudes are
\begin{eqnarray}
A_{+}&=&-i
{(\eta_{0}\cos{\psi})}^{1/2}{\Bigl[\eta_{0}^2\cos^2{\psi}+\mu_1(z-h)\Bigr]}^
{-1/4}\ ,
\\
A_{-}&=&{(\eta_{0}\cos{\psi})}^{1/2}{\Bigl[\eta_{0}^2\cos^2{\psi}+\mu_1(z-h)
\Bigr]}^{-1/4} \ .
\end{eqnarray}
Therefore the modified amplitudes $(\ref{modamplitudes})$ are
given by
\begin{equation}
g_0=-i\sqrt{2}{(\eta_0\cos{\psi})}^{1/2}
{\Bigl(\frac{1}{\mu_1}\Bigr)}^{1/6}, \qquad g_1=0 \ .
\end{equation}

It is then easy to check that the Kravtsov-Ludwig formula
coincides in the layer $0<z<h$ with the analytical solution of the
Dirichlet boundary value problem in the half space $z<h$. In fact,
by separation of variables we look for solutions of the form
$U(y,z)= \exp(i\kappa_0 y \sin\psi)\, u(z)$ , and it follows that
$u(z)$ satisfies the ordinary equation $u''(z) + \kappa_0^2
\bigl(\mu_{0}+\mu_{1}z -\sin^{2}\psi\bigr)\, u(z)=0$ which using the
changes of variables  $Z=\mu_{0}+\mu_{1}z -\sin^{2}\psi$ and
$x=Z\e^{1/3}\ , \ \ \e= (\mu_1/{\kappa_0})^2$ , is transformed to the Airy
equation
\begin{equation} \label{sclairy}
\e^2 {{u^{\e}}^{\prime\prime}(x)}+xu^{\e}(x)=0\, ,\quad x\in R \ .
\end{equation}

In the sequel we are interested for the high-frequency regime,
that is when $\e$ is small, and we study as a model problem the
geometrical optics of the semiclassical Airy equation
$(\ref{sclairy})$.

%%%%%%%%%%%%%%%%%%%%%%%%%%%%%%%%%%%%%%%%%%%%%%%%%%%%%%%%%%%%%%%%%%%%
\section{Geometrical optics for the semiclassical Airy equation}

According to the WKB method we are looking for asymptotic solution
of the semiclassical Airy equation
\begin{equation}
\e^2 {{u^{\e}}^{\prime\prime}(x)}+xu^{\e}(x)=0\, ,\quad x\in R 
\nonumber
\end{equation}
in the form
$$
u^{\e}(x)=A(x)\, e^{iS(x)/\e}
$$
where $S(x)$ is a real-valued phase and $A(x)$ is the
complex-valued principal amplitude (from now on we drop the
subscript in the principle amplitude $A_0$), solving the eikonal
equation
\begin{equation}\label{airyeikonal}
(S^\prime (x))^2=x
\end{equation}
and transport equation $(\ref{transport})$, respectively.

 The Hamiltonian $(\ref{hamiltonian})$ is given by
\begin{equation}
H(x,k)=\frac12 \bigl(k^2 -x \bigr) \ ,
\end{equation}
and the corresponding Hamiltonian system $(\ref{hamsys})$ has the
simple form
\begin{eqnarray}\label{hamsysairy}
\frac{dx}{dt}=H_k=k \ ,\quad \frac{dk}{dt}=-H_x=\frac{1}{2} \ .
\end{eqnarray}

We assume that the rays are launched from a point source at
$x=x_0$, where $x_0>0$, therefore $(\ref{hamsysairy})$ must
satisfy the initial conditions
\begin{equation}
x(0)=x_0 \ , \quad k(0)=k_0 \ .
\end{equation}

Solving the system $(\ref{hamsysairy})$ with the above initial
conditions we obtain the bicharacteristics
\begin{eqnarray}\label{airybich}
x(t;x_0,k_0)= t^2/4 +k_0 t +x_0 \ ,\quad k(t;x_0,k_0)=t/2 + k_0\ .
\end{eqnarray}

Now for constructing the rays, since $k_0$ must satisfy the
condition $H(x_0,k_0)=0$, we have $k_0 = \pm\, \sqrt{x_0}$ . The
positive sign ($k_0 = + \sqrt{x_0}$) corresponds to the ray
\begin{equation}\label{rayr}
x_{R}(t;x_0)= t^2/4 + t\sqrt{x_0}  +x_0
\end{equation}
moving from $x_0$ towards $x=+\infty$ , while the negative one
($k_0 = - \sqrt{x_0}$) to the ray
\begin{equation}\label{rayl}
x_{L}(t;x_0)= t^2/4 - t\sqrt{x_0}  +x_0
\end{equation}
moving from $x_0$ towards the turning point $x=0$ (which is the
caustic of the problem as we will see in the sequel).

The Jacobian of the right-moving ray
$x_R$ is
\begin{equation}\label{jacobr}
J_{R}(t;x_0)=\frac{\partial x_R}{\partial x_0}=1+ t/2\sqrt{x_0} \
,
\end{equation}
and it is always positive. However, the Jacobian of the left-moving ray $x_L$ ,
\begin{equation}\label{jacobl}
J_{L}(t;x_0)=\frac{\partial x_L}{\partial x_0}=1- t/2\sqrt{x_0} \
,
\end{equation}
vanishes for $t=t_c :=2\sqrt{x_0}$ corresponding to $x_L (t_c;
x_0)=0$ . Therefore, $x=0$ is the caustic for the left-moving ray
and it coincides with the turning point of the Airy equation.

We now observe that for any $x>x_0$ the equation $x_{R}(t;x_0)=x$
has the single solution $t_{R}=2(\sqrt{x} -\sqrt{x_0})$ , while for
$0<x<x_0$ the equation $x_{L}(t;x_0)=x$ has two solutions
\begin{eqnarray}\label{arrtimes}
t_{-}=2(\sqrt{x_0} -\sqrt{x}) \ ,\quad t_{+}=2(\sqrt{x_0}
+\sqrt{x}) \,
\end{eqnarray}
and the corresponding values of the Jacobian $J_L$ are
\begin{eqnarray}\label{arrjac}
J_{-}=\frac{\sqrt{x}}{\sqrt{x_0} }>0 \ ,\quad
J_{+}=-\frac{\sqrt{x}}{\sqrt{x_0} }<0 \ .
\end{eqnarray}

The arrival time $t_{-}$ and the Jacobian $J_{-}$ correspond to
the ray left-moving ray from the source, while $t_{+}$ and $J_{+}$
correspond to the ray reflected from the caustic $x=0$ .

Moreover, using the formula $(\ref{gphase})$ and imposing the
condition that the geometric phase of the rays emitted from the
source must vanish at the source point  (see Avila \& Keller
\cite{AK} for a detailed analysis of the geometrical optics with
point sources), we obtain the geometric phases
\begin{equation}\label{airyphases}
S_{\pm}(x)=\pm \frac{2}{3}x^{3/2} + \frac{2}{3}x_{0}^{3/2} \ , \ \
0<x<x_0 \ ,
\end{equation}
and
\begin{equation}
S_{R}(x)=\frac{2}{3}\bigl(\sqrt{x} -2\sqrt{x_0}\bigr)^{3} +
\frac{2}{3}x_{0}^{3/2} \ , \ \ x>x_0 \ .
\end{equation}

Note that $S_{R}(x_0)=0$ , $S_{-}(x_0)=0$ , that is the rays
emitted by the source satisfy the Avila-Keller condition, while
$S_{+}(x_0)=\frac{4}{3}x_{0}^{3/2}$ for the reflected ray and
$S_{+}(x=0)=S_{-}(x=0)=\frac{2}{3}x_{0}^{3/2}$ .

Concerning the geometry of the ray system, obviously the
bicharacteristics $(\ref{airybich})$ lie on the Lagrangian
manifold $\Lambda=\{(x,k): x=k^2\}$ since $k^2= (t/2 + k_0)^2=
x+(k_0^2 - x_0)=x$ , and for $k=S_{\pm , R}'(x)$ we have in fact
$H(x,k)=0$ (see Figure 2.3).
\newpage
\begin{figure}[htbp]
 \centering
  \includegraphics[scale=0.3]{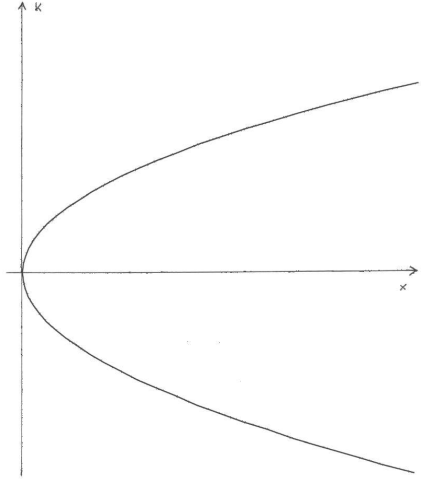}
    \caption{{\it Lagrangian manifold for the semiclassical Airy equation}}
\end{figure}

The principal amplitudes $A_{\pm}(x; x_0)$ of the WKB waves along
the rays in the region $0<x<x_0$ are found from
$(\ref{amplitude})$ using the corresponding values of the
Jacobian, and they are given by
\begin{eqnarray}\label{arrampl}
A_{-}=\frac{\alpha_0}{\sqrt{J_{-}} }= \frac{\alpha_0
x_0^{1/4}}{x^{1/4}} \ , \quad A_{+}=\frac{\alpha_0}{\sqrt{J_{+}}
}= -i\frac{\alpha_0 x_0^{1/4}}{ x^{1/4}} \ .
\end{eqnarray}
Then, the {\it multiphase WKB solution} in the region $0<x<x_0$ is
given by

\begin{eqnarray}\label{wkbairy}
u^{\e}_{WKB}(x) &=& A_{+}e^{iS_{+}(x)/\e} +
A_{-}e^{iS_{-}(x)/\e} \nonumber \\
&=& \alpha_0 x_0^{1/4}e^{i \frac{1}{\e}  \frac{2}{3}x_{0}^{3/2}}
\Bigl( -i x^{-1/4}e^{i \frac{1}{\e}  \frac{2}{3}x^{3/2}} +
x^{-1/4}e^{-i \frac{1}{\e}  \frac{2}{3}x^{3/2}} \Bigr) \ .
\end{eqnarray}
Here $\alpha_0$ is the {\it WKB amplitude of the wave at the
source} and it is equal to $\alpha_0 = e^{-1/4} x_0^{-1/2}/2$ .
This value follows from the asymptotics of the fundamental
solution of the semiclassical Airy equation which are presented in
Appendix A, in particular by comparing $(\ref {wkbairy})$ with
(\ref{wkbfundairy}). The same value would be derived by applying
the Avila-Keller technique for the WKB approximation of
fundamental solutions near the source point.

Finally the Kravtsov-Ludwig solution is found from the formula
$(\ref{klformula})$. To apply this formula we compute the
Kravtsov-Ludwig coordinates $(\ref{klphases})$ and the modified
amplitudes $(\ref{modamplitudes})$,
\begin{eqnarray}
\phi(x)&=&\frac12 \Bigl(S_{+}(x) + S_{-}(x)\Bigr)=
\frac{2}{3}x_{0}^{3/2}\\
\rho (x)&=&\Biggl[\frac34\biggl(S_{+}(x) +
S_{-}(x)\biggr)\Biggr]^{2/3}= x 
\end{eqnarray}
and
\begin{eqnarray}
g_0 (x)&=& \frac{1}{\sqrt 2}\, \rho^{1/4}(x)\Bigl(A_{+}(x) -i
A_{-}(x)
\Bigr)= -\frac{1}{\sqrt 2}\ e^{i\pi/4} x_0^{-1/4} \\
g_1(x)&=& \frac{1}{\sqrt 2}\, \rho^{1/4}(x)\Bigl(A_{+}(x) +i A_{-}(x)
\Bigr)=0 \ .
\end{eqnarray}

Then the {\it KL uniform solution} in the region $0<x<x_0$ is
given by
\begin{equation}\label{klairy}
u^{\e}_{KL}(x)=\pi^{1/2}e^{-i\pi/2}\Bigl(x_0^{-1/4}e^{i
\frac{1}{\e} \frac{2}{3}x_{0}^{3/2}}\Bigr)\e^{-1/6}
Ai\Biggl(-\frac{x}{\e^{2/3}}\Biggr) 
\end{equation}
and it coincides with the fundamental solution (\ref{innerairy})
derived in Appendix D.

%%%%%%%%%%%%%%%%%%%%%%%%%%%%%%%%%%%%%%%%%%%%%%%%%%%%%%%%%%%%%%%%%%%%%%%%%%%%
%%%%%%%%%%%%%%%%%%%%%%%%%%%%%%%%%%
%%%%%%%%%%%%%%%%%%%%%%%%%%%%%%%%%%%%%%%%%%%%%%%%%%%%%%%%%%%%%%%%%%%%%%%%%%%%
%%%%%%%%%%%%%%%%%%%%%%%%%%%%%%%%%%
\chapter{Wigner functions and its asymptotics}

\section{The Wigner transform and basic properties}

 For any smooth complex valued function $\psi(x)$ rapidly decaying at infinity, say
$\psi
\in \mathcal{S}(R)$ , the Wigner transform of $ \psi$ is a
quadratic transform defined by
\begin{equation}\label{wig}
W[\psi](x,k)=W(x,k)=\frac{1}{2\pi}\int_{R}{e}^{-ik
y}\, \psi(x+\frac{y}{2})\, \overline\psi(x-\frac{y}{2})\, dy
\end{equation}
where $\cpsi$ is the complex conjugate of $\psi$. The Wigner
transform is defined in phase space $R_{xk}$ , it is real, and it
has, among others, the following remarkable properties.

First, the integral of $W(x,k)$ wrt. $k$ gives the squared
amplitude (energy density) of $\psi$,
\begin{equation}
\int_{R}W(x,k) dk=|\psi(x)|^2 \ .
\end{equation}
In fact, we have
\begin{eqnarray}
\int_{R}W(x,k) dk&=&\frac{1}{2\pi}\int_{R} \int_{R}e^{-ik
y}\, \psi(x+\frac{y}{2})\, \cpsi(x-\frac{y}{2})\, dy\, dk \nonumber \\
&=&\int_{R}\left(\frac{1}{2\pi}\int_{R}e^{-ik
y}\, dk\right)\psi(x+\frac{y}{2})\, \overline{\psi}(x-\frac{y}{2})\, dy \nonumber \\
&=&\int_{R}\delta(y)\, \psi(x+\frac{y}{2})\, 
\overline{\psi}(x-\frac{y}{2})\, dy
\nonumber\\
&=&\psi(x)\, \overline{\psi}(x)=|\psi(x)|^2 .
\end{eqnarray}
where we have used the Fourier transform
$\delta(y)=\frac{1}{2\pi}\int_{R}e^{-iky}\, dk\, $ $\delta$ of Dirac's
measure.

Second, the first moment  of $W(x,k)$ wrt. to $k$ gives the
(energy) flux of $\psi$,

\begin{equation}
\int_{R}kW(x,k)\, dk=\frac{1}{2i}\Bigl (\psi(x)\, \overline{\psi^{'}}(x)-
\cpsi(x)\, \psi^{'}(x)\Bigr )={\mathcal F}(x) \ .
\end{equation}
In fact, we have
\begin{eqnarray}
\int_{R}kW(x,k)\, dk&=&\int_{R}\left(\frac{1}{2\pi}\int_{R} k
e^{-iky}\, dk
\right)\psi(x+\frac{y}{2}) \, \cpsi(x-\frac{y}{2})\, dy \nonumber\\
&=&-\frac{1}{i}\int_{R}\delta^{'}(y)\, \psi(x+\frac{y}{2})\, 
\cpsi(x-\frac{y}{2})\, dy \nonumber\\
&=&\frac{1}{i}\int_{R}\delta(y)\left(\frac12\psi^{'}(x+\frac y2)\, 
\cpsi(x-\frac{y}{2})-\frac12\overline{\psi^{'}}(x-\frac
y2)\, \psi(x+\frac y2)\right)\, dy \nonumber\\
&=&\frac{1}{2i}\left(
\cpsi(x)\, \psi^{'}(x)-\psi(x)\, \overline{\psi^{'}}(x)\right)\ .
\end{eqnarray}
The $x$ to $k$ duality in phase space can be recognized using the
alternative definition
\begin{equation}\label{wigf}
W(x,k)=\int_{R}e^{ipx}\, \widehat{\psi}(-k-\frac
p2)\, \overline{\widehat{\psi}}(-k+\frac p2)\, dp \ ,
\end{equation}
where $\widehat{\psi}(k)=\frac{1}{2\pi}\int_{R} e^{ikz}\, \psi(z)\, dz$
denotes the Fourier transform of $\psi$.

In fact, the definitions $(\ref{wig})$ and $(\ref{wigf})$ are
equivalent, since we have
\begin{eqnarray}
W(x,k)&=&\frac{1}{2\pi}\int_{R}e^{-ik
y}\, \psi(x+\frac{y}{2}) \, \cpsi(x-\frac{y}{2})\, dy  \nonumber\\
&=&\frac{1}{2\pi}\int_{R}e^{-iky}\int_{R}
e^{-iz(x+\frac{y}{2})}\, \widehat{\psi}(z)\, dz \overline{\int_{R}
e^{-iw(x-\frac{y}{2})}\, \widehat{\psi}(w)\, dw}\, dy \nonumber\\
&=&\frac{1}{2\pi}\int_{R}e^{-iky}\int_{R}
e^{-iz(x+\frac{y}{2})}\, \widehat{\psi}(z)\, dz \int_{R}
e^{iw(x-\frac{y}{2})}\, \overline{\widehat{\psi}}(w)\, dw\, dy \nonumber\\
&=&\int_{R}\int_{R}\left(\frac{1}{2\pi}\int_{R}e^{-iy(k+\frac{z}{2}+\frac{w}
{2})}\, dy\right)
e^{-i(z-w)x}\, \widehat{\psi}(z)\, \overline{\widehat{\psi}}(w)\, dz\, dw \nonumber\\
&=&\int_{R}\int_{R}\delta(k+\frac{z}{2}+\frac{w}{2})\, e^{-i(z-w)x}\, \widehat{\psi}(z)\, 
\overline{\widehat{\psi}}(w)\, dz\, dw \nonumber\\
&=&2\int_{R}e^{-i2(k+z)x}\, \widehat{\psi}(z)\, \overline{\widehat{\psi}}(-2k-z)\, dz
\nonumber\\
&=&\int_{R}e^{ipx}\, \widehat{\psi}(-k-\frac{p}{2})\,\overline{\widehat{\psi}}(-k
+\frac{p}{2})\, dp\ .
\end{eqnarray}

As we have explained in the previous chapter, in the case of high
frequency wave propagation, it is useful to use WKB wave functions
of the form
\begin{equation}\label{wkbfunc}
\psi^{\e}(x)= A(x)\, e^{i{S(x)}/{\e}} \ ,
\end{equation}
where $S(x)$ is a real-valued and smooth phase, and $A(x)$ is a
real-valued and smooth amplitude of compact support or at least
rapidly decaying at infinity. The Wigner transform of
$\psi^{\e}(x)$ is the Wigner function
\begin{equation}
W(x,k): = W[\psi^{\e}](x,k) = \frac{1}{2\pi}
\int_{R}e^{-iky}\, e^{\frac{i}{\e}S(x+\frac y2)}\, A(x+\frac
y2)\, e^{-\frac{i}{\e}S(x-\frac y2)}\, \overline{A}(x-\frac y2)\, dy \ ,
\end{equation}
but $W(x,k)$ does not converge  to a nontrivial
limit, as $\e \rightarrow 0$ . However, it can be shown that the
rescaled version of $W(x,k)$ , that we call the scaled Wigner
transform of $\psi^{\e}$,
\begin{equation}
W^{\e}(x,k)=\frac{1}{\e}W\Bigl(x,\frac{k}{\e}\Bigr) \
\end{equation}
converges weakly as $\e\rightarrow 0$  to the limit Wigner
distribution \cite{PR}, \cite{LP}
\begin{equation}\label{diracwig}
W^{0}(x,k)=|A(x)|^2\frac{1}{2\pi}\int_{R}
e^{-i(k-S'(x))y}\, dy=|A(x)|^2\delta(k-S'(x)) \ ,
\end{equation}
which is a Dirac measure concentrated on the Lagrangian manifold
$k=S'(x)$ , associated with the phase of the WKB  wavefunction, and
it is the correct weak limit of $W^{\e}$ (see, e.g., Lions \& Paul
\cite{LP}).

Indeed, proceeding formally, we rewrite $W^{\e}$ in the form
$$
W^{\e}(x,k)=\frac{1}{2\pi}\int_{R}e^{-iky}\, A(x+\frac{\e y}{2}
)\, \overline{A}(x-\frac{\e y}{2})\, e^{\frac{i}{\e}\left[ S(x+\frac {\e
y}{2})-S(x-\frac{\e y}{2} )\right]}\, dy \ ,
$$
and we expand in Taylor series about $x$ both $A(x \pm \frac{\e
y}{2} )$ and $S(x \pm \frac{\e y}{2} )$ . Then, we have
\begin{eqnarray*}
A(x+\frac{\e y}{2} )\, \overline{A}(x-\frac{\e y}{2})
&=&\left(A(x)+\frac{\e}{2}yA^{'}(x)+\dots\right)\left(\overline{A}(x)-
\frac{\e}{2}y\overline{A}^{'}(x)+\dots\right)\\
&=&A(x)\, \overline{A}(x)+ O(\e)\\
&=&|A(x)|^2 + O(\e) \ ,
\end{eqnarray*}
and
\begin{eqnarray*}
 S(x+\frac {\e y}{2})-S(x-\frac{\e y}{2}
 )&=&\left(S(x)+\frac{\e}{2}yS^{'}(x)+\frac{\e^2}{8}y^2S^{''}(x)+\dots
 \right)\\
 &-&
\left(S(x)-\frac{\e}{2}yS^{'}(x)+\frac{\e^2}{8}y^2S^{''}(x)-\dots
 \right)\\
 &=&\e y S^{'}(x) +O(\e^3) \ .
\end{eqnarray*}
Retaining only terms of order $O(1)$ in $A$ and $O(y)$ in $S$, and
integrating the expansion termwise we obtain that $W^{\e}(x,k)$
``converges" to $(\ref{diracwig})$.

More  precisely, if $Q$ is any test function in $\mathcal
S(R^2_{xk})$ , then
$$
\int_{R} \int_{R}Q(x,k)\, W^{\e}(x,k)\, dx\, dk\rightarrow
\int_{R}Q(x,S^{'}(x))\, |A(x)|^2dx \ .
$$
The above observations suggest that the scaled Wigner transform
\begin{eqnarray}\label{wige}
W^{\e}(x,k)&=&\frac{1}{\e}W\Bigl(x,\frac{k}{\e}\Bigr) \nonumber\\
&=&\frac{1}{2\pi}\int_{R}e^{-iky}\, \psi^{\e}(x+\frac{\e y}{2})\, 
\overline \psi^{\e}(x-\frac{\e y}{2})\, dy  \ ,
\end{eqnarray}
is the correct phase-space object for analyzing high frequency
waves.

%%%%%%%%%%%%%%%%%%%%%%%%%%%%%%%%%%%%%%%%%%%%%%%%%%%%%%%%%%%%%%%%%%%%%%%%%%%%
%%%%%%%%%%%%%%%%%%%%
%%%%%%%%%%%%%%%%%%%%%%%%%%%%%%%%%%%%%%%%%%%%%%%%%%%%%%%%%%%%%%%%%%%%%%%%%%%%
%%%%%%%%%%%%%%%%%%%%%

\section{Asymptotics of the Wigner function for a WKB wave function}

Consider now the  (scaled) Wigner function
\begin{eqnarray}
\we(x,k)=\frac{1}{\pi\e}\int_{R}^{}\ \psi^{\e}(x+\si)\, 
\overline{\psi}^{\e}(x-\si)\, \ e^{-\frac{i}{\e}2k\si}\, d\si
\end{eqnarray}
of the WKB wave function
\begin{equation}
\psi^{\e}(x)=A(x)\, e^{iS(x)/\e}
\end{equation}
where we assume that $A$, $S$ are smooth and real-valued, and $S'(x)$
is globally concave.

We want to construct an asymptotic expansion of $W^{\e}(x,k)$,
that is the oscillatory integral
\begin{equation}\label{wigwkb}
W^{\e}(x,k)=\frac{1}{\pi\e}\int_{R}^{}D(\si,x)\, 
e^{i\frac{1}{\e}F(\si,x,k)}\, d\si \ ,
\end{equation}
where
\begin{equation}
D(\si,x)=A(x+\si)A(x-\si)
\end{equation}
is the amplitude, and
\begin{equation}\label{wigphase}
F(\si,x,k)=S(x+\si)-S(x-\si)-2k\si
\end{equation}
is the Wigner phase. Asymptotics of such integrals are usually
constructed by applying the method of stationary phase.

For this purpose, we first compute the critical points of the
phase $F(\si;x,k)$, that is the roots of
\begin{equation}\label{alpha}
F_{\si}(\si;x,k)=S'(x+\si)+S'(x-\si)-2k=0
\end{equation}

\begin{figure}[ht]
\centering
\includegraphics[width=0.6\textwidth]{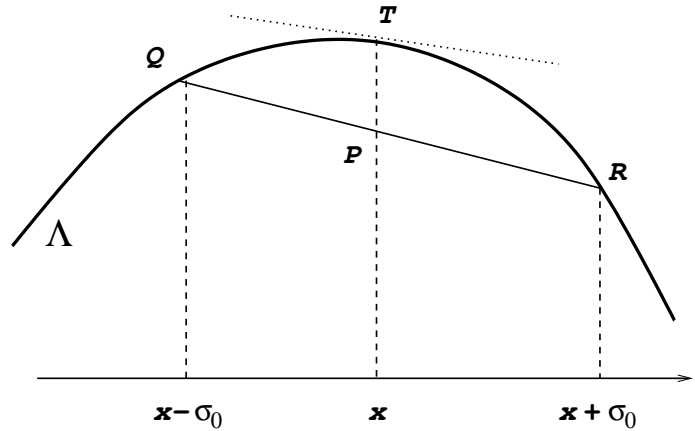}
\caption{{\it Berry's chord}}
\end{figure}

By the geometrical picture of Figure 3.1 (Berry's chord
construction; see the seminal paper by Berry \cite{Ber}), we
conclude that $(\ref{alpha})$ has a pair of symmetric roots $\pm
\si_{0}(x,k)$ such that the point $P=(x,k)$ be the middle of a
chord $QR$ having its ends on the Lagrangian (``manifold") curve
$\Lambda =\{k=S'(x)\}$ .

We observe that as $P$ moves toward $\Lambda$ the chord $QR$ tends
to the tangent of $\Lambda$ and $\si_{0}(x,k)\rightarrow 0$ .
Therefore, the two critical points of $(\ref{alpha})$ tend to
coalesce to the double point $\si=0$ as $P$ moves towards
$\Lambda$.

In fact, we have
\begin{equation}
F_{\si\si}(\si,x,k)=S''(x+\si)-S''(x-\si)
\end{equation}
and
\begin{equation}
F_{\si\si\si}(\si,x,k)=S'''(x+\si)+S'''(x-\si)
\end{equation}
and therefore
\begin{equation}
F_{\si\si}(\si=0,x,k)=0\, , \quad F_{\si\si\si}(\si=0,x,k)=2S'''(x)\not=0
\end{equation}
which assert that $\si=0$ is a double stationary point of $F$.

We would like therefore to apply a uniform stationary formula like
that derived in Appendix C which holds even when the stationary
points coalesce. For this we need to identify the parameter
$\alpha$ of the uniform stationary formula, which controls the
distance between the stationary points of the Wigner phase. In
order to do this we expand F about $\si=0$ , 
\begin{eqnarray*}
F(\si;x,k)&=& S(x)+\si S'(x)+\frac{\si^2}{2}S''(x)+\frac{\si^3}{6}
S'''(x)+...\\
&-& \biggl(S(x)-\si
S'(x)+\frac{\si^2}{2}S''(x)-\frac{\si^3}{6}S'''(x)+...\biggr)-2k\si\\
&=&-2(k-S'(x))\si+\frac{1}{3}S'''(x)\si^3+O(\si^5)\ .
\end{eqnarray*}

It becomes evident that for $P$ lying close enough to $\Lambda$, 
the parameter $\alpha$ has to be identified as
\begin{equation}
\alpha=\alpha(x,k):=k-S'(x)\ .
\end{equation}
Then, for any fixed $x$ we rewrite the Wigner phase $F$ in the
form
\begin{eqnarray}
F(\si;\alpha,x)&=& S(x+\si)-S(x-\si)-2\si(\alpha+S'(x))\nonumber\\
&=& \Bigl (S(x+\si)-S(x-\si)-2\si S'(x)\Bigr )-2\si\alpha \ ,
\end{eqnarray}
and we have
\begin{equation}
F_{\si\si}(\si=0;x,k)=0\, , \quad  F_{\si\si\si}(\si=0;x,k)=2S'''(x)\, , \quad
F_{\si\alpha}(\si=0;\alpha,x)=-2 \ .
\end{equation}
These are exactly the conditions in Appendix C which are needed
for applying the uniform asymptotic formula.

Then, the asymptotic formula of Appendix C  ($B_0$ vanishes since $D$ is even wrt
$\si$) gives the following
approximation of the Wigner function
\begin{eqnarray}\label{fullwigairy}
W^{\e}(x,k) \approx 2\pi A_{0}(x,k) {\e}^{2/3} Ai\Bigl( - \e^{2/3} \xi(x,k)\Bigr)
\end{eqnarray}
where
\begin{equation}\label{fulldelta}
A_{0}(x,k)=2^{-1/2}\xi^{1/4}\frac{2D(\si_{o}(x,k),x)}{\mid
F_{\si\si}(\si_{0}(x,k);x,k)\mid ^{1/2}} \ ,
\end{equation}
\begin{eqnarray}\label{fullzeta}
\xi(x,k) =  \Biggl[ \frac32 \Bigl(S(x+ \si_0) -S(x-\si_0) -2k\si_0
\Bigr)\Biggr]^{2/3} 
\end{eqnarray}
and
\begin{eqnarray}\label{gamma}
F_{\si \si}(\si=\si_{0};x,k)=S''(x+\si_0)-S''(x-\si_0)<0 \ .
\end{eqnarray}

We now further approximate  the various quantities entering
$(\ref{fullwigairy})$ as $\alpha\rightarrow 0$ . First of all, in
this approximation
\begin{eqnarray}\label{zeta}
\xi &\approx
&-F_{\si\alpha}\Biggl(\frac{F_{\si\si\si}}{2}\Biggr)^{-1/3}\alpha\nonumber\\
&=& 2(k-S'(x))\Biggl(\frac{2S'''(x)}{2}\Biggr)^{-1/3}\nonumber\\
&=&2(S'''(x))^{-1/3}(k-S'(x))\ .
\end{eqnarray}
Furthermore we approximate
\begin{equation}\label{beta} F_{\si
\si}(\si=\si_{0};x,k)=S''(x+\si_0)-S''(x-\si_0)\approx
2\si_0(x,k)S'''(x)
\end{equation}
and since $\si_{0}$ is approximated by
\begin{eqnarray}
\si_{0}&\approx &(-2F_{\si\si\si}\ F_{\si\alpha}\
\alpha)^{1/2}(F_{\si\si\si})^{-1}\nonumber\\
&= &(-2\cdot  2 S'''(x)(-2)(k-S'(x))^{1/2}(2S'''(x))^{-1}\nonumber\\
&=& \Biggl(\ \frac{2}{S'''(x)}(k-S'(x))\ \Biggr)^{1/2} 
\end{eqnarray}
when $\alpha\rightarrow 0$ , we have
\begin{equation}
F_{\si\si}(\si_{0},x,k)\approx 2\si_{0}S'''(x) =2 \Biggl[\
\frac{2(k-S'(x))}{S'''(x)}\ \Biggr]^{1/2}S'''(x) \ ,
\end{equation}
and then
\begin{equation}
\frac{\xi^{1/4}}{\mid F_{\si\si}\mid ^{1/2}} =\Biggl(\frac{2}{\mid
S'''(x)\mid} \Biggr)^{1/3}\frac{1}{2^{1/2+1/3}} \quad  .
\end{equation}

Moreover, since $D(\si_{o};x)=D(-\si_{0};x)$ , it follows
\begin{equation}\label{epsilon}
B_{0}=0
\end{equation}
and using the approximation $(\ref{zeta})$ of $\xi$, we have
\begin{equation}\label{delta}
A_{0}=\frac{1}{2^{1/3}}\Biggl(\frac{2}{\mid S'''(x)\mid}
\Biggr)^{1/3}D(\si_{0}(x,k),x)\ .
\end{equation}

Finally, using $(\ref{fullwigairy})$ with
$$F(\si=0;x,k)=0\ ,\quad  B_{0}=0$$
and $\xi$, $A_{0}$ given by ${(\ref{zeta})}$, ${(\ref{delta})}$,
respectively, we arrive to the Airy-type approximation of
${(\ref{wigwkb})}$,
\begin{eqnarray}\label{sclwigairy}
\tw(x,k)&\approx& \frac{2^{2/3}}{\e^{2/3}}\Biggl(\frac{2}{\mid
S'''(x)\mid} \Biggr)^{1/3}D(\si_{o}(x,k),x)\cdot \nonumber \\
&\cdot& Ai\Biggl(-\frac{2^{2/3}}{\e^{2/3}}\Biggl(\frac{2}{\mid
S'''(x)\mid} \Biggr)^{1/3}(k-S'(x))\Biggr)
\end{eqnarray}
which is Berry's semiclassical approximation of $W^{\e}$, and we
call it the semiclassical Wigner function (of the WKB function).
Note that $(\ref{sclwigairy})$ is an approximation of
$(\ref{fullwigairy})$ which holds locally near the manifold
($\alpha=k-S'(x)$ very small).

%%%%%%%%%%%%%%%%%%%%%%%%%%%%%%%%%%%%%%%%%%%%%%%%%%%%%%%%%%%%%%%%%%%%%%%%%%%%%%%%%%%%%%%%%%%%%%%%%%%%%%
%%%%%%%%%%%%%%%%%%%%%%%%%%%%%%%%%%%%%%%%%%%%%%%%%%%%%%%%%%%%%%%%%%%%%%%%%%%%%%%%%%%%%%%%%%%%%%%%%%%%%%%%

\chapter{Wignerization of two-phase WKB solutions}

In this chapter we study the structure of the Wigner transform of
wave functions whose high-frequency asymptotics are described by
two-phase WKB solutions, which is typical for wave fields around
fold caustics. In order to start understanding the related
asymptotic mechanisms, we first investigate the Wigner transform
of the Airy function and its relation to the Wigner transform of
the WKB asymptotic solution of the semiclassical Airy equation.

More precisely, we first compute the exact Wigner transform
$\we_{Ai}$ (see $(\ref{wigfundairy})$ below) of the fundamental
solution
\begin{equation}\label{fundairy}
u^{\e}(x)= \pi^{1/2} e^{-i \pi/2} \Bigl(x_0^{-1/4}e^{i
\frac{1}{\e} \frac{2}{3}x_{0}^{3/2}} \Bigr) \e^{-1/6}
Ai\left(-\e^{-2/3}x \right) \ , \nonumber
\end{equation}
(cf eq. $(\ref{innerairy})$) of the semiclassical Airy equation.
Recall that the Kravtsov-Ludwig formula $(\ref{klairy})$ coincides
with $(\ref{fundairy})$ in this case.

In the sequel we compute the asymptotics of the Wigner transform
of the WKB approximation $(\ref{wkbairy})$ of the fundamental
solution, using the semiclassical Wigner function developed in
previous chapter, and we show that this approximation coincides
with the exact Wigner transform $(\ref{wigfundairy})$ below.

\section{Wigner transform of the Airy function}

We start with the integral representation
\begin{equation}\label{airy}
\pse(x):=Ai(-\epsilon^{-2/3}x)=\frac{1}{2\pi}\int_{R}^{ } e^{i(\frac{\rho^3}{3}-\epsilon^{-2/3}x\rho)}\,  d\rho \ ,
\end{equation}
of the Airy function.
 The scaled Wigner transform $\we$ of $\pse$ is given by
\begin{eqnarray}
\we[\pse](x,k)=\frac{1}{\pi\e}\int_{R}^{ }\
\pse(x+\si)\, \overline{\psi}^\e(x-\si) \, e^{-\frac{i}{\e}2k\si}\, d\si \,
\end{eqnarray}
and therefore, if we substitute $(\ref{airy})$ and we put
$\lambda=\epsilon^{-2/3}$, we have
\begin{eqnarray}
\we[\pse](x,k)&=&(2\pi)^{-3}
\int_{R}^{ }\int_{R}^{ }e^{i \frac13 (\rho^3
- \sigma^3)} \, e^{-i\lambda x (\rho-\sigma)}
\int_{R}^{ }e^{-i\left(k+\lambda
\e(\rho+\sigma)/2\right)\tau }\, d\tau \, d\rho\,  d\sigma
\nonumber\\
&=&(2\pi)^{-2}\, 2\e^{-1/3}\int_{R}^{ }\int_{R}^{ }e^{i
\frac13 (\rho^3 - \sigma^3)}\, e^{-i\lambda x
(\rho-\sigma)}\, \delta\left(\rho+\sigma+2k\e^{-1/3}\right)\, d\rho \, 
d\sigma \nonumber \ ,
\end{eqnarray}
where Dirac's mass is expressed through the Fourier transform
\begin{equation}
\delta(z)=\frac{1}{2\pi}\int_{R}^{ }e^{-iz\tau}\, d\tau \
.
\end{equation}

On the support of  Dirac's mass
$\sigma=-(\rho+2\epsilon^{-1/3}k)$, and $\rho-\sigma=2(\rho+u)$
with $u=\epsilon^{-1/3}k$ , so we have
$$
\rho^3 - \sigma^3= 2\rho^3 + 6u\rho^2 +12 u^2 \rho +8u^3 \ .
$$
After some straightforward algebra, and by the linear change
$\rho=2^{-1/3}r -u$ , we obtain

\begin{equation}
\we[\pse](x,k)=\frac{2^{2/3}}{2\pi}\, 
\e^{-1/3}\int_{R}^{ } e^{i\Bigl(\frac{\tau^3}{3}
+ 2^{-1/3}\left(2u^2-2\lambda x\right)\tau \Bigr)}\, d\tau\
,
\end{equation}
which by the integral representation ($\ref{airy}$) gives the
Wigner transform of the Airy function
\begin{equation}
\we[\pse](x,k)=\frac{1}{2^{1/3} \e^{1/3} \pi }\, 
 Ai\Bigl(2^{2/3}\epsilon^{-2/3}(k^2-x)\Bigr)\ . 
\end{equation}

Then, the Wigner transform of the fundamental solution
$(\ref{fundairy})$ is given by

\begin{equation}\label{wigfundairy}
\we_{Ai}(x,k):=\we[u^\e](x,k)=\frac{1}{2^{1/3}
\e^{2/3}}\, x_{0}^{-1/2}
 Ai\Bigl(2^{2/3}\epsilon^{-2/3}(k^2-x)\Bigr) \ .
\end{equation}

By employing the asympotics of the Airy function we see that in the
interior of the
Lagrangian manifold $x>k^2$, $\we[u^\e]$ oscillates at the scale $\e$,
\begin{equation}
\we[u^\e](x,k) \approx
\frac{1}{\sqrt{2\pi}}\ \epsilon^{-1/2}x_{0}^{-1/2}{(x-k^2)}^{-1/4}
\cos{\Biggl(\frac{4}{3\e}(x-k^2)^{3/2}-\frac{\pi}{4}\Biggr)} \ ,
\end{equation}
while at the exterior of the manifold (which is connected to the shadow
region)
$x<k^2$, $\we[u^\e]$ decays exponentially
\begin{equation}
\we[u^\e](x,k) \approx
\frac{1}{2^{3/2}\sqrt{\pi}}\ \epsilon^{-1/2}x_{0}^{-1/2}{(k^2-x)}^{-1/4}
e^{-\frac{4}{3\e}(k^2-x)^{3/2}} \ .
\end{equation}

This asymptotic picture suggests the existence of a transition boundary
layer with
thickness $O(\e^{2/3})$ around the Lagrangian manifold $x=k^2$ , inside
which the
wave field is described by Airy structure  and where most of the energy of
the wave
field is concentrated.

The weak limit of $\we[u^\e]$ as $\epsilon\rightarrow 0$ , is computed by
the formula
\begin{equation}
\frac{1}{\e}F\left(\frac{z}{\e}\right) \ \rightarrow \delta(z)\int_{R}^{
}F(y)\, dy \ , \ \ \epsilon\rightarrow 0 \ ,
\end{equation}
and, since
$$
\int_{R}^{ }Ai(y)dy =1 \ ,
$$
it is given by
\begin{eqnarray}
\weo(x,k)=x_{0}^{-1/2}2^{-1/3}\delta\left(2^{2/3}(k^2-x)\right)=
\frac{1}{2 x_{0}^{1/2}}\, \delta(k^2-x)\ .
\end{eqnarray}

Note that in the illuminated zone $x>0$ ,
\begin{equation}
\weo(x,k)=\frac{1}{2 x_{0}^{1/2}}\, 
\delta(k^2-x)=\frac{1}{4 x_{0}^{1/2}x^{1/2}}
\Bigl(\delta(k-x^{1/2}) +
\delta(k+x^{1/2})   \Bigr) \ ,
\end{equation}
that is, $\weo$ splits to two Dirac masses supported on the branches $k=\pm\, 
x^{1/2}$ . This splitting fails on the caustic $x=0$ , while in the shadow
zone $x<0$
the limit Wigner $\weo$ is weakly zero since $k^2-x >0$ .

On the other hand we can compute the limit Wigner distribution of the
(two-phase)
WKB expansion of the fundamental solution $(\ref{wkbairy})$ of the
semiclassical
Airy equation (cf also (\ref{wkbfundairy})),
\begin{eqnarray*}
u^{\e}_{WKB}(x) &=& A_{+}(x)\, e^{iS_{+}(x)/\e} +
A_{-}(x)\, e^{iS_{-}(x)/\e} \nonumber \\
&=& \alpha_0 x_0^{1/4}\, e^{i \frac{1}{\e}  \frac{2}{3}x_{0}^{3/2}}
\Bigl( -i x^{-1/4}e^{i \frac{1}{\e}  \frac{2}{3}x^{3/2}} +
x^{-1/4}\, e^{-i \frac{1}{\e}  \frac{2}{3}x^{3/2}} \Bigr) \ ,
\end{eqnarray*}
By $(\ref{diracwig})$ we have the weak limits
$$
\we\Biggl[A_{\pm}(x)e^{iS_{\pm}(x)/\e} \Biggr] \ \rightarrow \ {\mid
A_{\pm}(x)
\mid}^{2} \delta(k-S'_{\pm}(x)) \ , \quad \e \rightarrow 0 \ .
$$
Moreover, the cross Wigner transform
\begin{eqnarray}
&\we\biggl[A_{+}(x)\, e^{iS_{+}(x)}, A_{-}(x)\, e^{iS_{-}(x)}\biggr]= \nonumber \\
&=(\pi\e)^{-1} \int_{R}^{}A_{+}(x+ \si)\, \overline{
A}_{-}(x-\si)\, e^{i\frac{1}{\e}\bigl(S_{+}(x+\si)
-S_{+}(x-\si)-2k\si\bigr)}\, d\si \ ,
\end{eqnarray}
converges weakly to zero, since following the same reasoning as for proving
$(\ref{diracwig})$ by expanding the phase and  the amplitude in Taylor series wrt. $\si$, we
get in
front of the integral the oscillatory term
$$
e^{i{\frac{1}{\e}}\bigl(S_{\pm}(x)-S_{\mp}(x)\bigr)} \ ,
$$
which weakly tends to zero as $\e \rightarrow 0$, for $S'_{+}(x)\neq - S'_{-}(x)$ .

It then follow that
\begin{eqnarray}\label{wignerwkb}
\we[u^{\e}_{WKB}](x,k) \ \rightarrow {\mid A_{+}(x) \mid}^{2}
\delta(k-S'_{+}(x)) +{\mid A_{-}(x) \mid}^{2}
\delta(k-S'_{-}(x)) \nonumber \\
=\frac{1}{4 x_{0}^{1/2}x^{1/2}}
\biggl(\delta(k-x^{1/2})  +
\delta(k+x^{1/2})   \biggr) \ ,
\end{eqnarray}
and, as it is anticipated for $\e \rightarrow 0$ , we derive that
\begin{equation}
\we[u^{\e}_{WKB}](x,k) \rightarrow  \weo(x,k)\ .
\end{equation}

%%%%%%%%%%%%%%%%%%%%%%%%%%%%%%%%%%%%%%%%%%%%%%%%%%%%%%%%%%%%%%%%%%%%%%%%%%%%%%%%%%
\section{Wigner transform of the WKB expansion for the Airy
equation}

We have already constructed the WKB solution of the Airy equation
in the form
\begin{equation}
u^{\epsilon}_{WKB}(x)=A_{+}(x)\,
e^{\frac{i}{\epsilon}S_{+}(x)}+A_{-}(x)\,
e^{\frac{i}{\epsilon}S_{-}(x)}\ ,
\end{equation}
where the phases $S_{\pm}$ and the amplitudes $A_{\pm}$ are given
by $(\ref{airyphases})$ and $(\ref{arrampl})$,
$$
S_{\pm }(x)=\pm \frac{2}{3}x^{3/2}+\frac{2}{3}{x_{0}}^{3/2} \ ,
$$
and
$$
A_{+}(x)=(-i)\frac{1}{2}x^{-1/4}e^{-i\pi/4}{x_0}^{-1/4} \ , \ \
A_{-}(x)=\frac{1}{2}x^{-1/4}e^{-i\pi/4}{x_0}^{-1/4} \ .
$$

The scaled Wigner transform of $u^{\epsilon}_{WKB}$ is
 given by
\begin{equation}
W_{WKB}^\e(x,k)=\frac{1}{\pi\e}\sum_{\ell=1}^{4}\int_{R}{}
D_{\ell}(\sigma;x)\ e^{\frac{i}{\e}F_{\ell}(\sigma;x,k)}\, d\sigma=
\sum_{\ell=1}^{4}W^\e_{\ell}(x,k)
\end{equation}
where
\begin{equation}
W^\e_{\ell}(x,k)=\int_{R}^{} D_{\ell}(\sigma;x)\,
e^{\frac{i}{\e}F_{\ell}(\sigma;x,k)}\, d\sigma \, , \quad \ell=1,..,4 \
.
\end{equation}
The amplitudes and phases of the above four Wigner integrals are
given by
\begin{eqnarray}
D_1(\sigma;x)&=&A_{+}(x+\sigma)\overline{A}_{+}(x-\sigma) \nonumber \\
D_2(\sigma;x)&=&A_{-}(x+\sigma)\overline{A}_{-}(x-\sigma)  \nonumber \\
D_3(\sigma;x)&=&A_{+}(x+\sigma)\overline{A}_{-}(x-\sigma) \nonumber \\
D_4(\sigma;x)&=&A_{-}(x+\sigma)\overline{A}_{+}(x-\sigma) 
\end{eqnarray}
and
 \begin{eqnarray}
F_1(\sigma;x,k)&=&S_{+}(x+\sigma)-S_{+}(x-\sigma)-2k\sigma\\
F_2(\sigma;x,k)&=&S_{-}(x+\sigma)-S_{-}(x-\sigma)-2k\sigma \\
F_3(\sigma;x,k)&=&S_{+}(x+\sigma)-S_{-}(x-\sigma)-2k\sigma  \\
F_4(\sigma;x,k)&=&S_{-}(x+\sigma)-S_{+}(x-\sigma)-2k\sigma \ .
\end{eqnarray}

In the sequel we compute the stationary-phase asymptotic
expansions of the Wigner integrals $W^\e_{\ell}$ , using either the
standard or the uniform formula according to the structure of the
stationary points in each case.

%%%%%%%%%%%%%%%%%%%%%%%%%%%%%%%%%%%%%%%%%%%%%%%%%%%%%%%%%%
%%%%%%%%%%%%%%%%%%%%%%%%%%%%%%%%%%%%%%%%%%%%%%%%%%%%%%%%%%

\subsection{Stationary points of the Wigner phases}

In the sequel we compute the stationary points of the Wigner
phases in the illuminated area $x>0$ , since the real-valued phases
$S_{\pm}(x)$ of the WKB solution have been computed only in the
illuminated region. It turns out that all real stationary points,
which give the main asymptotic contribution to the Wigner
integrals, lie in the area $|\sigma|<x$ in which the Wigner phases
are real. Outside this area, the stationary points are imaginary
and their contribution to the Wigner integrals is exponentially
small.

{\bf {Stationary points of the Wigner phase $F_1(\sigma;x,k)$ .}}
The critical points of $F_1(\sigma;x,k)$ are given by
\begin{equation}\label{symeq}
 {F_1}
_{\sigma}(\sigma;x,k)=S_{+}^{\prime}(x+\sigma)+S_{+}^{\prime}(x-\sigma)-2k=0
\ ,
\end{equation}
that is
\begin{equation}\label{stateqa}
(x+\sigma)^{1/2}+(x-\sigma)^{1/2}-2k=0 \ .
\end{equation}

For $k<0$ we see that the phase $F_1(\sigma;x,k)$ has not critical
points since $(x \pm \sigma)^{1/2}>0$ , while for $k>0$ we can see
that the roots of $(\ref{symeq})$ appear in symmetric pairs. In
fact, if we set $\sigma=\sigma_R+i\sigma_I$ and substitute into
(\ref{stateqa}) we have
\begin{equation}\label{complexeq}
(\sigma_R-2k^2)^2-{\sigma_I}^2+2i\sigma_I(\sigma_R-2k^2)=4k^2(x-\sigma_R)-4k^2\sigma_Ii
\ ,
\end{equation}
and equating the real and imaginary parts to zero, we obtain
\begin{equation}\label{stsya}
(\sigma_R-2k^2)^2-{\sigma_I}^2-4k^2(x-\sigma_R)=0
\end{equation}
\begin{equation}\label{stsyb}
\sigma_I\sigma_R=0 \ .
\end{equation}
Thus we must consider the following cases.

\noindent
 {\bf{Case 1 :}} $\sigma_I=0$ . Then,
$\sigma=\sigma_R\in R$ and $(\ref{stsya})$ gives
\begin{equation}
(x^2-\sigma^2)^{1/2}= 2k^2-x \ ,
\end{equation}
which implies the restriction $x\le 2k^2$ . From the last equation
we find that the real critical points of $F_1$ are
\begin{equation}\label{stpoint}
\sigma(x,k)=\pm 2|k|(x-k^2)^{1/2}=:\pm \sigma_0(x,k)\ .
\end{equation}
Therefore, in the region $k^2<x\leq 2k^2$   $F_1$  has two real
stationary points, the $\pm \sigma_0(x,k)=\pm 2|k|(x-k^2)^{1/2}$ ,
which coalesce to $\sigma(x,k)=0$ on the upper branch of the
Lagrangian manifold $x=k^2$ . Now since
\begin{equation}
{F_1} _{\sigma\sigma}(\sigma=\pm
\sigma_0;x,k)=\frac{1}{2}\frac{(x\mp \sigma_0)^{1/2}-(x\pm
\sigma_0)^{1/2}}{(x^2-\sigma_0^2)^{1/2}}\not=0 \ ,
\end{equation}
and
\begin{equation}
{F_1} _{\sigma\sigma}(\sigma=0;x,k)=0 \ , \quad {F_1}
_{\sigma\sigma\sigma}(\sigma=0;x,k)=-\frac{1}{2}x^{-3/2}\not=0 \ ,
\end{equation}
it turns out that the points $\pm \sigma_0(x,k)$ are simple and
the point $\sigma(x,k)=0$  formatted by the coalescence of $\pm
\sigma_0(x,k)$ is double. In fact, by Berry's chord construction
we see that, as we move toward the Lagrangian manifold
$\{(x,k):k=S_{+}^{\prime}(x)\}$ the chord tends to the tangent of
the manifold, and the critical points $\pm \sigma_{0}$ tend to the
double stationary point $\sigma=0$ .

\noindent
 {\bf{Case 2 :}} $\sigma_I\not=0$ . Then from (\ref{stsyb}) we have
$$
\sigma_I=\pm 2|k|(k^2-x)^{1/2} \ ,
$$
and therefore for $x<k^2$ $F_1$ has simple imaginary stationary
points $\sigma(x,k)=\pm 2|k|i(k^2-x)^{1/2}$ ,  since in this region
\begin{equation}
{F_1} _{\sigma\sigma}(\sigma=\pm i \sigma_0;x,k)\not=0 \ .
\nonumber
\end{equation}

\noindent {\bf{Stationary points of the Wigner phase
$F_2(\sigma;x,k)$ .}} The critical points of $F_2$ are given by the
equation
\begin{eqnarray*}
{F_2}
_{\sigma}(\sigma;x,k)=S_{-}^{\prime}(x+\sigma)+S_{-}^{\prime}(x-\sigma)-2k=0
\end{eqnarray*}
that is,
\begin{eqnarray}\label{stateqb}
(x+\sigma)^{1/2}+(x-\sigma)^{1/2}+2k=0 \ .
\end{eqnarray}

For $k>0$ the equation (\ref{stateqb}) has no solution since $(x
\pm \sigma)^{1/2}>0$ . For $k<0$ we set $\sigma=\sigma_R+i\sigma_I$
 into (\ref{stateqb}) and we obtain again the
system (\ref{stsya}), $(\ref{stsyb})$ . Therefore, the critical
points of $F_2$ are $\sigma(x,k)=\pm 2|k|(x-k^2)^{1/2}=\pm
\sigma_0$ , when $k^2\leq x\leq 2k^2$ and $\sigma(x,k)=\pm
2|k|i(k^2-x)^{1/2}=\pm i\sigma_0$ , when $x< k^2$ .

At any point $(x,k)$ in $k^2< x\leq 2k^2$ there exist two simple
stationary points, $\pm \sigma_0$. For fixed $(x,k)$ moving
towards $S_{-}^{\prime}(x)=-\sqrt{x}$ we see again that the
critical points coalesce to the double point $\sigma=0$ . Finally,
in the region $x<k^2$ we have two simple imaginary stationary
points, $\pm i \sigma_0$ , which again coalesce to $\sigma=0$ on
the lower branch of the Lagrangian manifold $x=k^2$ .

\noindent
 {\bf{Stationary points of the Wigner phase
$F_3(\sigma;x,k)$ .}} The critical points of the phase 
$F_3(\sigma;x,k)$ are given by the equation
\begin{eqnarray}
(x+\sigma)^{1/2}-(x-\sigma)^{1/2}-2k=0\ .
\end{eqnarray}

In this case we find that the solutions of the above equation are
$\sigma=\pm 2|k|(x-k^2)^{1/2}=\pm \sigma_0$ for $x \geq 2k^2$ , and
$\sigma=\pm 2|k|i(k^2-x)^{1/2}=\pm i\sigma_0$  for $x<k^2$ . By
geometrical considerations similar to Berry's chord construction,
we see that for fixed $(x,k)$ with $x>2k^2$ ,  the stationary point
is $\sigma=+\sigma_0$ if $k>0$, while if $k<0$ the stationary
point is $\sigma=-\sigma_0$ . These stationary points are always
simple since
\begin{equation}
{F_3} _{\sigma\sigma}(\sigma=\pm
\sigma_0;x,k)=\frac{1}{2}\frac{(x\mp \sigma_0)^{1/2}+(x\pm
\sigma_0)^{1/2}}{(x^2-\sigma_0^2)^{1/2}}\not=0 \ ,
\end{equation}
and
\begin{equation}
{F_3} _{\sigma\sigma}(\sigma=\pm i \sigma_0;x,k)=\frac{1}{2}\frac{(x\mp i
\sigma_0)^{1/2}+(x\pm i \sigma_0)^{1/2}}{(x^2+\sigma_0^2)^{1/2}}\not=0 \ .
\end{equation}

Note that ${F_3} _{\sigma\sigma}(\sigma=\pm \sigma_0;x,k)$ becomes
infinite for $x=2k^2$ . 

Finally, it is important to observe that in this case there are no
stationary points in the region $k^2< x< 2k^2$ .

\noindent
 {\bf{Stationary points of the Wigner phase
$F_4(\sigma;x,k)$ .}} In this case the critical points are
$\sigma=\pm \sigma_0$ when $x\geq 2k^2$  and $\sigma=\pm i \sigma_0$
when $x<k^2$ . Here, for fixed $x$ with $x>2k^2$ , the stationary
point is $\sigma=-\sigma_0$ for $k>0$ , and $\sigma=+\sigma_0$ for
$k<0$ . Again, the stationary points $\pm \sigma_0$ and $\pm i
\sigma_0$ are simple.

For easier consideration of the structure of the stationary
points, the results of the above computations are tabulated in the
following two tables, where we have set
$\sigma_0=\sigma_0(x,k)=2\ |k|\mid x-k^2\mid ^{1/2}$ .
\begin{center}
\begin{tabular}[!hbp]{|c|c|c|c|c|c|c|c|c|}\hline
$F_i,k$ & $(-\infty,-\sqrt{x})$ & $(-\sqrt{x},-\sqrt{x/2})$ &
$(-\sqrt{x/2},\sqrt{x/2})$ & $(\sqrt{x/2},\sqrt{x})$ & $(\sqrt{x},\infty)$
\\
\hline $F_1$ & no s.p. & no s.p. &no s.p.  & $\sigma=\pm \sigma_0$ &
$\sigma=\pm i\sigma_{0}$\\
 & & & &simples &simples \\
\hline $F_2$ & $\sigma=\pm i\sigma_{0} $ & $\sigma=\pm \sigma_0$ & no s.p.
& no s.p. & no s.p. \\
 &simples & simples  & & & \\
\hline $F_3$ & $\sigma=\pm i\sigma_{0}$& no s.p. & $\sigma=\pm \sigma_0$ &
no s.p. & $\sigma=\pm i\sigma_{0}$\\
 &simples&  &simple & &simples \\
\hline  $F_4$ &  $\sigma=\pm i\sigma_{0}$ & no s.p.  & $\sigma=\pm
\sigma_0$ &no s.p.& $\sigma=\pm i\sigma_{0}$\\
&simples &  &simple & &simples \\
\hline
\end{tabular}
\end{center}
\begin{table}[h]
\caption{\it{A. Stationary points between the parabolas} }
\end{table}

\begin{center}
\begin{tabular}[h]{|c|c|c|c|c|c|c|c|}\hline
$F_i,k$ & $k=-\sqrt{x}$ & $k=-\sqrt{x/2}$ & $k=\sqrt{x/2}$ & $k=\sqrt{x}$  \\
\hline $F_1$ & no s.p.& no s.p.& $\sigma=\pm \sigma_0$ & $\sigma=0$\\
& & & simples & double\\
\hline $F_2$ & $\sigma=0$ & $\sigma=\pm \sigma_0$  &  no s.p.& no s.p. \\
& double  & simples &  &  \\
\hline $F_3$ & no s.p.& $\sigma=-\sigma_0$& $\sigma=+\sigma_0$&  no s.p.\\
& & simple& simple&  \\
\hline  $F_4$ &  no s.p. & $\sigma=+\sigma_0$ & $\sigma=-\sigma_0$ & no s.p. \\
& & simple& simple& \\
\hline
\end{tabular}
\end{center}
\begin{table}[h]
\caption{\it{B. Stationary points on the parabolas}}
\end{table}

\subsection{Asymptotics of the  diagonal Wigner function}

For constructing the asymptotics of the diagonal Wigner functions
$W_1^{\epsilon}$ , $W_2^{\epsilon}$ , we first observe that the
asymptotic contribution to $W_1^{\epsilon}$ comes from the
stationary points in the region $x \leq 2k^2$ , $k>0$ , while the
contribution to $W_2^{\epsilon}$ comes from the stationary points
in the region $x \leq 2k^2$ , $k<0$ , since there are no stationary
point of the corresponding Wigner phases outside from the above
indicated regions, respectively.
\\
We therefore consider the following two regions.

{\bf{Region 1 :}} $k^2\leq x\leq 2k^2$

Because $\sigma =0$ is a double stationary point for the diagonal
Wigner phases on the Lagrangian manifold $x=k^2$ (see Tables 4.1,
4.2) we need to apply the uniform approximation
formula $(\ref{fullwigairy})$ of Section 3.2. .

For the integral $W_1^{\epsilon}(x,k)$  we choose the parameter
$\alpha=\alpha(x,k):=k-\sqrt{x}$ . Expanding $F_1$ in Taylor series
about $\sigma=0$ , we have
\begin{eqnarray}\label{taylorF1}
{F_1}(\sigma;\alpha,x)=-2\alpha \sigma-\frac{1}{12}x^{-3/2}\sigma^3
\end{eqnarray}
and ${F_1}_{\sigma\alpha}(\sigma=0,\alpha=0,x)=-2\neq 0$ , and we
easily see that all other conditions for the validity of
$(\ref{fullwigairy})$ also hold. For applying this formula, we need
to compute the unknowns $\xi$, $A_0$, $B_0$, and
$F_1(\sigma=0,\alpha=0,x)$ .

From $(\ref{fullzeta})$ we have
\begin{eqnarray*}
\xi(x,k)=\Biggl[\frac{3}{4}\biggl(F_{1}(\sigma_0)-F_{1}(-\sigma_0)\biggr )\Biggr]^{2/3}=
\Biggl[\frac{3}{4}\biggl (4\frac{2}{3}(x-k^2)^{3/2}\biggr )\Biggr ]^{2/3}=2^{2/3}(x-k^2)\
,
\end{eqnarray*}
and from (\ref{fulldelta}) we get
\begin{eqnarray*}
A_0=2^{7/6}
(x-k^2)^{1/4}(x^2-{\sigma_0}^2)^{-1/4}x^{3/4}{x_{0}}^{-1/2}{\sigma_0}^{-1/2}
\ .
\end{eqnarray*}
Since $(x^2-{\sigma_0}^2)^{1/2}=2k^2-x$ , we rewrite $A_0$ in the
form
\begin{eqnarray}
A_0=2^{-4/3}k^{-1/2}x^{3/4}{x_0}^{-1/2}(2k^2-x)^{-1/2} \ .
\end{eqnarray}
and from $(\ref{taylorF1})$ we also have
$F_1(\sigma=0,\alpha=0,x)=0$ .

Thus $(\ref{fullwigairy})$ leads to the following approximation
formula of $W_1^{\epsilon}$ for small $\epsilon$,
\begin{eqnarray}
W_1^{\epsilon}(x,k)\approx
\frac{1}{2\sqrt{x_0}}\biggl(\frac{2}{\epsilon}\biggr )^{2/3}
Ai\Biggl(\biggl(\frac{2}{\epsilon}\biggr)^{2/3}(k^2-x)\Biggr)
=:\widetilde{W}_{1}^{\epsilon}(x,k) \ .
\end{eqnarray}

Similarly for the integral $W_2^{\epsilon}(x,k)$ we choose the
parameter $\alpha:=k+\sqrt{x}$ and compute that,
\begin{eqnarray*}
\xi=2^{2/3}(x-k^2) \ ,\quad
A_0=2^{-4/3}(-k)^{-1/2}x^{3/4}{x_0}^{-1/2}(2k^2-x)^{-1/2}
\end{eqnarray*}
and
\begin{eqnarray*}
 F_2(\sigma=0,\alpha=0,x)=0
\end{eqnarray*}
then from (\ref{fullwigairy}) as $\epsilon\rightarrow 0$ we have,
\begin{eqnarray}
W_2^{\epsilon}(x,k)\approx
\frac{1}{2\sqrt{x_0}}\biggl(\frac{2}{\epsilon}\biggr)^{2/3}
Ai\Biggl(\biggl(\frac{2}{\epsilon}\biggr)^{2/3}(k^2-x)\Biggr)
=:\widetilde{W}_{2}^{\epsilon}(x,k)\ .
\end{eqnarray}

{\bf{Region 2 :}} $x\leq k^2$

In this region the imaginary stationary points $\sigma(x,k)=\pm
2|k|i(k^2-x)^{1/2}=\pm i \sigma_0(x,k)$ coalesce for $x=k^2$ to
the double point $\sigma=0$ .

For the integral $W_1^{\epsilon}(x,k)$ , as in the case of real
stationary points, we apply the uniform stationary formula
$(\ref{fullwigairy})$ for $k>0$ , with small parameter
$\alpha:=k-\sqrt{x}$ . In this case we have
\begin{eqnarray*}
\xi(x,k)=\Biggl[\frac{3}{4}\biggl (F_{1}(i\sigma_0)-F_{1}(-i\sigma_0)\biggr )\Biggr]^{2/3}=\Biggl[\frac{3}{4}
\biggl (-4i\frac{2}{3}(k^2-x)^{3/2}\biggr )\Biggr ]^{2/3}=-2^{2/3}(k^2-x)
\end{eqnarray*}
and
\begin{eqnarray*}
A_0=2^{-4/3}{x_0}^{-1/2}(-1)^{1/4}i^{-1/2} \ ,  \quad
F_1(\sigma=0,\alpha=0,x)=0
\end{eqnarray*}
Thus $(\ref{fullwigairy})$ gives the approximation
\begin{eqnarray}
W_1^{\epsilon}(x,k)\approx
\frac{1}{2\sqrt{x_0}}\biggl(\frac{2}{\epsilon}\biggr )^{2/3}
Ai\Biggl(\biggl(\frac{2}{\epsilon}\biggr)^{2/3}(k^2-x)\Biggr )
=:\widehat{W}_{1}^{\epsilon}(x,k)\ .
\end{eqnarray}
%\ddot
Similarly, for the integral $W_2^{\epsilon}(x,k)$ we choose the
parameter $\alpha:=k+\sqrt{x}$ , with $k<0$ and we have
\begin{eqnarray*}
\xi(x,k)=-2^{2/3}(k^2-x)
\end{eqnarray*}
and
\begin{eqnarray*}
A_0=2^{-4/3}{x_0}^{-1/2}(-1)^{1/4}i^{-1/2} \ , \quad
F_2(\sigma=0,\alpha=0,x)=0
\end{eqnarray*}
Thus $(\ref{fullwigairy})$ gives the approximation
\begin{eqnarray}\label{cwigner2}
W_2^{\epsilon}(x,k)\approx
\frac{1}{2\sqrt{x_0}}\biggl(\frac{2}{\epsilon}\biggr )^{2/3}
Ai\Biggl(\biggl(\frac{2}{\epsilon}\biggr)^{2/3}(k^2-x)\Biggr)
=:\widehat{W}_{2}^{\epsilon}(x,k) \ .
\end{eqnarray}

Note that the symbols $\widehat{W}_{1}^{\epsilon}$, 
$\widehat{W}_{2}^{\epsilon}$ denote the approximations of
${W_1}^{\epsilon}$ , ${W_2}^{\epsilon}$ in the region $x\leq
k^2$ , although they are formally the same with
$\widetilde{W}_{1}^{\epsilon}$, $\widetilde{W}_{2}^{\epsilon}$ which
denote the corresponding approximations in the region $k^2\leq
x\leq 2k^2$ with $k>0$ , $k<0$ , respectively.

\subsection{Asymptotics of the  non-diagonal Wigner functions}

Recall from Table 4.1 that the non-diagonal Wigner phases $F_3$,
$F_4$ have two simple real stationary points $\pm \sigma_0(x,k)$ in
the region $x\geq 2k^2$ , and also two imaginary stationary points $\pm
i \sigma_0(x,k)$ in the region $x<k^2$ . Recall also that these
phases have no stationary points in the intermediate region
$k^2\leq x\leq 2k^2$ .

We therefore consider the following two regions and we compute the
asymptotics of  $W_3^{\epsilon}$ , $W_4^{\epsilon}$ by applying
the standard stationary phase formula since the stationary points
are simple and never coalesce to a double point due to the lack of
stationary points in the intermediate region. The resulting
asymptotic formulae are

{\bf Region 1 :} $x\geq 2k^2$

\begin{eqnarray}
W_3^{\epsilon}(x,k)\approx-\frac{i{x_0}^{-1/2}}{2^{3/2}\pi^{1/2}\epsilon^{1/2}}\, (x-k^2)^{-1/4}e^{i\pi/4}e^{i\frac{4}{3\epsilon}(x-k^2)^{3/2}}
=:\widetilde{W}_{3}^{\epsilon}(x,k)
\end{eqnarray}
and
\begin{eqnarray}
W_4^{\epsilon}(x,k)\approx\frac{i{x_0}^{-1/2}}{2^{3/2}\pi^{1/2}\epsilon^{1/2}}\, (x-k^2)^{-1/4}e^{-i\pi/4}e^{-i\frac{4}{3\epsilon}(x-k^2)^{3/2}}
=:\widetilde{W}_{4}^{\epsilon}(x,k)\ .
\end{eqnarray}

It is important here to observe that in the region $x\geq 2k^2$  the
non-diagonal Wigner functions $W_3^{\epsilon}$ , $W_4^{\epsilon}$
are the asymptotic approximations of the expression
$$
\frac{1}{2\sqrt{x_0}}\biggl(\frac{2}{\epsilon}\biggr)^{2/3}
Ai\Biggl(\biggl(\frac{2}{\epsilon}\biggr)^{2/3}(k^2-x)\Biggr)
$$
and we can therefore substitute this expression in place of them.

{\bf Region 2 :} $x<k^2$

\begin{eqnarray}\label{cwigner3}
W_3^{\epsilon}(x,k)\approx\widehat{W}_{3}^{\epsilon}(x,k)
:=\frac{-i^{1/2}{x_0}^{-1/2}}{2^{5/2}\pi^{1/2}\epsilon^{1/2}}\, (x^2+{\sigma_0}^2)^{-1/4}\frac{1}
{\mid {F_3}_{\sigma\sigma}(i\sigma_0)\mid ^{1/2}}\cdot \nonumber\\
\cdot
\Bigl[e^{\frac{i}{\epsilon}F_{3}(i\sigma_0)+i\pi/2}+e^{\frac{i}{\epsilon}F_{3}(i\sigma_0)+i3\pi/2}
+e^{\frac{i}{\epsilon}F_{3}(-i\sigma_0)+i\pi/2}+
e^{\frac{i}{\epsilon}F_{3}(-i\sigma_0)+i3\pi/2}\Bigr] \ .
\end{eqnarray}
and
\begin{eqnarray}\label{cwigner4}
W_4^{\epsilon}(x,k)\approx\widehat{W}_{4}^{\epsilon}(x,k)
:=\frac{-i^{1/2}{x_0}^{-1/2}}{2^{5/2}\pi^{1/2}\epsilon^{1/2}}\, (x^2+{\sigma_0}^2)^{-1/4}\frac{1}
{\mid {F_4}_{\sigma\sigma}(i\sigma_0)\mid ^{1/2}}\cdot\nonumber\\
\cdot
\Bigl[e^{\frac{i}{\epsilon}F_{4}(i\sigma_0)+i\pi/2}+e^{\frac{i}{\epsilon}F_{4}(i\sigma_0)+i3\pi/2}
+e^{\frac{i}{\epsilon}F_{4}(-i\sigma_0)+i\pi/2}+
e^{\frac{i}{\epsilon}F_{4}(-i\sigma_{0})+i3\pi/2}\Bigr] \ .
\end{eqnarray}

From the last two equations we see that
\begin{equation}
W_3^{\epsilon}(x,k) + W_4^{\epsilon}(x,k)\approx\
\widehat{W}_{3}^{\epsilon}(x,k)+ \widehat{W}_{4}^{\epsilon}(x,k)=0 \ .
\end{equation}
This means that in the region $x<2k^2$ the contribution of the
non-diagonal Wigner functions is asymptotically negligible.

It is however important to stress here that the computation of the
Wigner function in the high-frequency regime through any direct
solution of the Wigner equation which will be presented in the
next chapter, would face severe difficulties because of the
(essentially cancelling) oscillations of the terms $W_3^{\epsilon}$ ,
$W_4^{\epsilon}$ which appear outside the Lagrangian manifold.

Therefore, combining the  asymptotic expansions derived in the
various regions we find that the leading order approximation of
the Wigner transform
 $W_{WKB}^\e(x,k)=W^\e[u^{\epsilon}_{WKB}](x,k)$ , where $u^{\epsilon}_{WKB}$ is the WKB
approximation of the fundamental solution of the semiclassical
Airy equation, is given for every $(x>0, k)$ by
\begin{equation}
W_{WKB}^\e(x,k)\approx \widetilde{W}_{WKB}^{\epsilon}(x,k)=
\frac{1}{2\sqrt{x_0}}\biggl(\frac{2}{\epsilon}\biggr)^{2/3}
Ai\Biggl(\biggl(\frac{2}{\epsilon}\biggr)^{2/3}(k^2-x)\Biggr) \ .
\end{equation}

We observe that this approximation coincides with  the exact
Wigner transform $(\ref{wigfundairy})$ of the fundamental solution
of the semiclassical Airy equation,
\begin{equation}
\widetilde{W}_{WKB}^{\epsilon}(x,k)\equiv W_{Ai}^\e(x,k) \ ,
\end{equation}
 and it therefore is meaningful on the caustic $x=0$ , in the sense
 that it can provide the correct amplitude of the wavefunction
 there.

In the sequel, for easy reference,  we collectively use the term
WKB-Wigner transform,  for the approximation $\widetilde{W}_{WKB}^\e$ .

%%%%%%%%%%%%%%%%%%%%%%%%%%%%%%%%%%%%%%%%%%%%%%%%%%%%%%%%%%%%%%%%%%%%%%%%%%%%%%%%%%%%%%%%%%%%%%%
\subsection{$k-$integration
of the asymptotics of the WKB-Wigner transform}

The amplitude $|\psi^{\epsilon}(x)|^2$ of a wavefunction, for any
fixed $x$, is given by the $k-$ integral of its Wigner function
\begin{equation}
|\psi^{\epsilon}(x)|^2 = \int_{R}^{} W^{\epsilon}(x,k)\, dk \nonumber
\end{equation}
and therefore, in general, for small $\epsilon$, we expect that
\begin{equation}
|u^{\epsilon}(x)|^2 \approx \int_{R}^{}\widetilde
W_{WKB}^{\epsilon}(x,k)\, dk \nonumber
\end{equation}
where $\widetilde W_{WKB}^{\epsilon}$ is the WKB-Wigner transform.

Of course, in our particular example  of the fundamental solution
$(\ref{fundairy})$  of the semiclassical Airy equation, it turns out
that the integration of the WKB-Wigner function gives the exact
amplitude of the fundamental solution, as the WKB-Wigner function
coincides with the exact one.

In fact, using the formula \cite{VS}
\begin{eqnarray}
\int_{-\infty}^{\infty}Ai(r_{1}k^2+r_{2}k+r_3)\
dk=\frac{2\pi}{\sqrt{r_{1}}}\frac{1}{2^{1/3}}\
{Ai}^2\Biggl(-\frac{r_{2}^2-4r_{1}r_{3}}{4^{4/3}r_{1}}\Biggr)\,  , \quad
r_1>0
\end{eqnarray}
with $r_1=(2/\epsilon)^{2/3}$ , $r_2=0$ , $r_3=-(2/\epsilon)^{2/3}
x$ , we obtain
\begin{equation}\label{amplfund}
|u^{\epsilon}(x)|^2 =\int_{R}^{}
\frac{1}{2\sqrt{x_0}}\biggl(\frac{2}{\epsilon}\biggr)^{2/3}
Ai\Biggl(\biggl(\frac{2}{\epsilon}\biggr)^{2/3}(k^2-x)\Biggr)\, dk =
\frac{\pi}{\sqrt{x_0}\epsilon^{1/3}}\ {Ai}^2 (-\epsilon^{-2/3}x) \
,
\end{equation}
which is the anticipated result as the fundamental solution is
given by
$$
u^{\e}(x)= C \pi^{1/2}x_0^{-1/4} \e^{-1/6} Ai(-\epsilon^{-2/3}x) \ ,
\ \ C= e^{-i\pi/2}e^{i \frac{1}{\e} \frac{2}{3}x_{0}^{3/2}} \ .
$$
The constant $C$ cannot be computed from the Wigner function. Indeed,
the first moment of the Wigner function is
\begin{equation}\label{fluxfund}
\hspace{-0.2cm}
\epsilon\,  Im\biggl(\frac{d}{dx}\, u^{\e}(x)\, \bar{u}^{\e}(x)\biggr)
 =
\int_{R}^{} kW^{\epsilon}(x,k)\, dk = \int_{R}^{} k
\frac{1}{2\sqrt{x_0}}\biggl(\frac{2}{\epsilon}\biggr)^{2/3}
Ai\Biggl(\biggl(\frac{2}{\epsilon}\biggr)^{2/3}(k^2-x)\Biggr)\, dk=0 \, .
\end{equation}
If we write
$u^{\e}$ in the form $u^{\e}=\alpha^{\epsilon}(x)\, e^{i\phi^{\epsilon}(x)}$, then $(\ref{fluxfund})$ implies that
$\frac{d}{dx} \phi^{\epsilon}(x)=0$, that is $\phi^{\epsilon}(x)=$ const.,
and obviously $C=e^{i\phi^{\epsilon}}$, $|C|=1$. However, it is possible to fix this
constant by going back to the
semiclassical Airy equation.

%%%%%%%%%%%%%%%%%%%%%%%%%%%%%%%%%%%%%%%%%%%%%%%%%%%%%%%%%%%%%%%%%%%%%%%%%%%%%%%%%%%%%%%%%%%%%%%%%%%%%
%%%%%%%%%%%%%%%%%%%%%%%%%%%%%%%%%%%%%%%%%%%%%%%%%%%%%%%%%%%%%%%%%%%

\subsection{The stationary Wigner equation}

Along the same lines as for time-dependent Schr\"odinger equation
(see, e.g., \cite{LP}; also \cite{Tat}), we can show for the
homogeneous Helmholtz equation
\begin{equation}
\e^2 {{u^{\e}}^{\prime\prime}(x)}+ \eta^2(x) u^{\e}(x)=0,\quad
x\in R \ . \nonumber
\end{equation}
(by formally dropping the time derivative and identifying $V(x)=-
\eta^2(x)/2$ in the time-dependent Schr\"odinger equation), that
the Wigner transform $\fe(x,k)$ of the wave function $u^{\e}(x)$
must satisfy stationary Wigner equation
\begin{equation}\label{statwigneq}
 k\partial_{x}\fe + 1/2 \bigl(\eta^2(x)\bigr)' \partial_{k}\fe = -1/2
\sum_{m=1}^{\infty} \alpha_m
\e^{2m}\bigl(\eta^2(x)\bigr) ^{(2m+1)} (x) \partial_{k}^{2m+1} \fe
(x,k) \ .
\end{equation}
Note that in the formal limit $\e =0$ the series disappears and
the limit Wigner distribution $f^{0}$ satisfies the classical
Liouville equation
\begin{equation}
k\partial_{x}\fe + 1/2 \bigl(\eta^2(x)\bigr)'
\partial_{k}\fe =0 \ .
\end{equation}

In the special case of the Airy equation, $\eta^2(x)=x$, the
stationary Wigner equation $(\ref{statwigneq})$ takes the very
simple form
\begin{equation}\label{airywigneq}
 k\partial_{x}\fe + 1/2 \partial_{k}\fe = 0 \ .
\end{equation}
It is important to note that in this case the Wigner equation
coincides with the limiting Liouville equation,and we easily find
that its solution is given by
\begin{equation}\label{soliouv}
\fe (x,k)= F^{\e}(x-k^2) \,
\end{equation}
where $F^{\e}(z)$ is an arbitrary differentiable function, which
may contain $\e$ as parameter.

We observe that the Wigner transform $(\ref{wigfundairy})$ of the
semiclasical Airy function indeed is of the form (\ref{soliouv}),
and therefore is a solution of the stationary Wigner equation
$(\ref{airywigneq})$, as it should be. However, $F^{\e}$ (Airy, in
our case) cannot be derived from $(\ref{airywigneq})$. This can be
done either by using the $*-$equation derived by Moyal \cite{Mo}
and Baker \cite{Ba} or, alternatively, employing the so called
deformation quantization procedure; see e.g., Zachos \cite{Za}),
but we will not enter this subject here.

Finally, we must emphasize that we cannot derive a pure Wigner
equation for the inhomogeneous Helmholtz equation, because in this
case the wave function itself appears in the right-hand side of
$(\ref{statwigneq})$ (see, e.g., Castella et al \cite{BCKP}; also
\cite{Za}), and therefore the influence of a point source cannot
be studied by directly solving the stationary Wigner equation.
This makes our approach of wignerization of a WKB solution to
derive a uniform solution in phase space, a promising way for
constructing solutions, or at least appropriate anzatz, for
kinetic equations.

\subsection{Remarks on the wignerization of a general two-phase WKB solution}

We first observe that in the case of semiclassical Airy equation,
by using the Airy phases $(\ref{airyphases})$,
$$
S_{\pm}(x)=\pm \frac{2}{3}x^{3/2} + \frac{2}{3}x_{0}^{3/2} \ ,
$$
the Kravtsov-Ludwig coordinates $(\ref{klphases})$ are
\begin{equation}
\phi(x)= \frac{2}{3}x_{0}^{3/2} \ ,  \ \ \ \phi'(x)=0
\end{equation}
and
\begin{equation}
\rho(x)=\Biggl[\frac34\biggl(S_{+}(x) - S_{-}(x)\biggr)\Biggr]^{2/3}
= x \ ,
\end{equation}
and they obviously satisfy $(\ref{klphasesyst})$. Also in the case
of a general smooth refraction index $\eta(x)$ , we obtain the same
results since the geometric phases  have the form
\begin{equation}
S_{\pm}(x)=\pm \int_{x_0}^{x}\eta^{2}(z)dz + S_{0}(x_0) \ ,
\end{equation}
and the whole procedure of the wignerization of the WKB solution
can be repeated along the same lines as for the semiclassical Airy
equation.

In the more general case of, for example, a two dimensional
problem with a smooth caustic (fold), the Kravtsov-Ludwig
coordinates $\phi$ arises naturally from the consideration of the
stationary points of the diagonal Wigner phases
\begin{eqnarray*}
F_1(\sigma;x,k)&=&S_{+}(x+\sigma)-S_{+}(x-\sigma)-2k\sigma\\
F_2(\sigma;x,k)&=&S_{-}(x+\sigma)-S_{-}(x-\sigma)-2k\sigma
\end{eqnarray*}
since the point $\sigma=0$ is stationary when
$$
k=\frac{1}{2} \Bigl(S'_{+}(x)+S'_{+}(x)\Bigr)=:\phi'(x) \ .
$$
However, the role of the second Kravtsov-Ludwig coordinate $\rho$,
is not obvious before we pass to local (tangent and normal)
coordinates at the caustic, which  are typically used in boundary
layer analysis (see, e.g., the detailed exposition in the book by
Babich \& Kirpichnikova \cite {BaKi}). In these  local coordinates
coordinates, the semiclassical Airy equation appears again in a
way quite similar with that in the model example presented in
Section 2.4., and then the proposed wignerization process can be
applied.

\chapter{Conclusions}

We have studied the asymptotic expansion of the Wigner transform (wignerization) of
the two-phase WKB solution of the
semiclassical Airy equation, combining uniform and standard stationary phase
approximations for two kinds of Wigner
integrals in various regions of the phase space (``surgery" of asymptotics).

The diagonal Wigner integrals in the illuminated zone have been locally handled
in the same way as in the derivation of the semiclassical Wigner function for
single-phase WKB functions. It turned out that this
derivation holds true in regions near the two branches of the folded Lagrangian
manifold $(x=k^2)$ , which are comprised between
the manifold and the
conjugate curve $(x=2k^2)$ , defined as the locus of points beyond which Berry's chord
does not exist any more.

The non-diagonal Wigner integrals have been handled in the illuminated region by the
standard stationary phase formula
including the contribution of two conjugate imaginary points appearing at points of
phase space at the exterior of the Lagrangian manifold.
However, although important for understanding the structure of the wignerized WKB
solution $W_{WKB}^\e$ , these imaginary points do not
contribute asymptotically to the approximation $\widetilde{W}_{WKB}^\e$ .
On the  other hand, the real symmetric stationary points arising at points of phase
space at the interior of the conjugate curve,
 offer oscillatory contributions
which can be recognized as the high-frequency asymptotics of $W_{Ai}^\e$ , an
observation which is decisive in the ``surgery" of asymptotic formulae
in order to show that $\widetilde{W}_{WKB}^\e$ is finally expressible in the Airy form
$$
\widetilde{W}_{WKB}^{\epsilon}(x,k)\equiv W_{Ai}^\e(x,k)=\frac{1}{2^{1/3}
\e^{2/3}}x_{0}^{-1/2}
 Ai\Bigl(2^{2/3}\epsilon^{-2/3}(k^2-x)\Bigr) \ .
$$
Although the whole process has been done in the illuminated zone, all formulae can
be meaningfully extended in the shadow zone, and
then they predict there the anticipated exponentially decaying wave fields away from
the caustic. Such an extension is plausible
in the light of complex geometric optics (see, e.g., Chapman et al \cite{CLOT}, for a
recent review),
and it has been probably used for first time in the works by J.B. Keller (see, e.g.,
\cite{SK}).

The analysis performed for the semiclassical Airy equation suggests how someone can
wignerize fold caustics  in two
dimensional propagation, by introducing local caustic coordinates which reveal the
essential Airy structure of the problem,
but the details of such a computation are long and still to be worked out.

Finally, the observation that the wignerized WKB solution $\widetilde{W}_{WKB}^\e$
satisfies the stationary Wigner equation
(although almost trivially in our model problem), suggests that, in the general
case, the wignerized two-phase WKB solution is
expected to be a formal asymptotic solution of the Wigner equation (a fact which has
been already
confirmed for the single-phase case in \cite{FM1}), which could be a fruitful way to
handle multiphase wave-kinetic equations.

%%%%%%%%%%%%%%%%%%%%%%%%%%%%%%%%%%%%%%%%%%%%%%%%%%%%%%%%%%%%%%%%%%%%%%%%%%%%%%%%%%%%%%%%%%%%%%
%%%%%%%%%%%%%%%%%%%%%%%%%%%%%%%%%%%%%%%%%%%%%%%%%%%%%%%%%%%%%%%%%%%%%%%%%%%%%%%%%%%%%%%%%%%%%%%
\appendix

\chapter{Stationary-phase method}
\begin{lemma}\label{statphlemma}
\begin{eqnarray*}
J&=&\int_{0}^{\infty}t^{\gamma}\, e^{i\nu t^{p}}\, dt\nonumber \\
&=&{\left(\frac{1}{|\nu |}\right)}^{\frac{\gamma+1}{p}}\ \frac{\Gamma
(\frac{\gamma+1}{p})}{p}\ e^{i\frac{\pi}{2p}(\gamma+1)sgn\nu }
\end{eqnarray*}
where $\gamma$ and $\nu$ are real constants, $\gamma>-1$ and $p$ is a positive integer.
\end{lemma}

A heuristic analysis of the leading term of the asymptotic expansion of Fourier type
integrals closely follows Laplace's method
(see also \cite{BH}).
Consider the integral
\begin{equation}\label{integral}
I(\lambda)=\int_{a}^{b}f(t)\, e^{i\lambda\phi(t)} \,  dt
\end{equation}
and suppose that $f\in C[a,b]$ while $\phi\in C^2[a,b]$, real-valued function.
Suppose further that $t=c$ is the only point in $[a,b]$ where $\phi'(t)$
vanishes and $\phi''(c)\not =0$. We rewrite $I(\lambda)$ as
\begin{eqnarray*}
I(\lambda)=e^{i\lambda\phi(c)}\int_{a}^{b}f(t)\, e^{i\lambda(\phi(t)-\phi(c))}\,  dt.
\end{eqnarray*}
The main contribution to the integral $(\ref{integral})$
comes from a small neighborhood of $c$.
Then we expect that the large $\lambda$ behavior of $(\ref{integral})$
is given by
\begin{eqnarray*}
\int_{c-r}^{c+r}f(c)\, e^{i\lambda[\phi(c)+\frac{(t-c)^2}{2}\phi''(c)]}\,  dt
\end{eqnarray*}
where $r$ is small but finite. To evaluate this integral, we let
\begin{eqnarray*}
\mu\tau^2=(t-c)^2 \ \frac{\phi''(c)}{2} \ \lambda,\quad  \textrm{or} \quad 
\tau=(t-c)\sqrt{\frac{|\phi''(c)|\lambda}{2}}
\end{eqnarray*}
where $\mu=sgn\phi''(c)$. Then the above integral becomes
\begin{eqnarray*}
f(c)e^{i\lambda\phi(c)}\sqrt{\frac{2}{|\phi''(c)|\lambda}}\int_{-r\sqrt{\lambda|\phi''(c)|/2}}^{r\sqrt{\lambda|\phi''(c)|/2}}e^{i\mu\tau^2}
d\tau
\end{eqnarray*}
As $\lambda\rightarrow \infty$ the last integral reduces to $\int_{-\infty}^{\infty}
e^{i\mu\tau^2} d\tau$, which can be evaluated
exactly
\begin{eqnarray*}
\int_{-\infty}^{\infty} e^{i\mu\tau^2} d\tau=2\int_{0}^{\infty} e^{i\mu\tau^2}
d\tau=\sqrt{\pi}e^{\frac{i\pi\mu}{4}}
\end{eqnarray*}
where we have used the lemma \ref{statphlemma}
with $\gamma=0$, $p=2$, and $\nu =\mu$ (recall that $\Gamma(1/2)=\sqrt{\pi}$ ).
Hence our formal analysis suggests that

\begin{equation}
I(\lambda)\approx e^{i\lambda\phi(c)+i\mu
\pi/4}f(c)\left[\frac{2\pi}{\lambda|\phi''(c)|}\right]^{1/2}
\end{equation}
as $\lambda\rightarrow \infty$, where $\mu=sgn\phi''(c)$.

%%%%%%%%%%%%%%%%%%%%%%%%%%%%%%%%%%%%%%%%%%%%%%%%%%%%%%%%%%%%%%%%%%%%%%%%%%%%%%%%%%%%%%%%%
%%%%%%%%%%%%%%%%%%%%%%%%%%%%%%%%%%%%%%%%%%%%%%%%%%%%%%%%%%%%%%%%%%%%%%%%%%%%%%%%%%%%%%%%%
\chapter{Sketch-proof of the Proposition $\ref{reprfold}$}

The idea of the proof is due to Chester, Friedman and Ursell
\cite{CFU}, who worked out the analytic rather the smooth case.

The starting point is the following lemma.
\begin{lemma}\emph{(Whitney's lemma)}
\label{whitney} Let $f$ be a smooth even function on the real
line. Then, there exists a smooth function $g$ on the real line
such that $f(x)=g(x^2)$ . If $f$ depends smoothly on a set  of
parameters, $g$ can be chosen so that it depends smoothly on the
same set of parameters.
\end{lemma}

Now, let us assume that the Lagrangian submanifold has a fold
point at $(\textbf{x}_0, \textbf{p}_0) \in M\times R^n$ . We assume
for simplicity that $M \subset R^n_{\bf{x}}$ , and that
$\textbf{x}_0 =0$. Let $S=S(\textbf{x}, \xi)$ on $M \times R$ be a
phase function parametrizing $\Lambda$ in the neighborhood of
$\textbf{x}_0$ , i.e., $(\textbf{x}, \textbf{k}= \nabla S) \in
\Lambda$ for $\textbf{x}$ near $\textbf{x}_0$ , and let
$$
C=\{(\textbf{x}, \xi) | \partial_{\xi}S(\textbf{x}, \xi)=0\} \ ,
$$
be the critical set of $\Lambda$. The caustic $\Sigma(\Lambda^n)$
consists of those points in $C$ where $\partial_{\xi}S(\textbf{x},
\xi)= \partial^{2}_{\xi}S(\textbf{x}, \xi)=0$ . Without loss of
generality, we may assume that the point in $C$ corresponding to
$(\textbf{x}_{0}, \textbf{p}_{0})$ is the origin. Then, the
following lemma holds.

\begin{lemma}
There exist smooth functions $v_0(\bf{x})$ and
$\rho(\bf{x})$ on $M$, and $\zeta(\bf{x}, \xi)$ on $M
\times R$ ,  such that
\begin{equation}\label{rhovfunc}
\frac{\zeta ^3}{3}-\rho \zeta + v_0 = S, \quad \frac{\partial
\zeta}{\partial \xi} >0,
 \quad {and} \quad \zeta^2 - \rho=0  \quad {for} \quad \ ({\bf{x}}, \xi)\in C \ .
\end{equation}
\end{lemma}
\begin{proof}

First us prove the assertion for the special case when the base
manifold $M$, is one dimensional $(n=1, \thinspace \textbf{x}
=x)$. The assumption that the origin is a fold point of $C$ means
that
$$
\frac{\partial S}{\partial \xi}=\frac{\partial^2 S}{\partial
\xi^2}=0 \quad {and} \quad  \frac{\partial^2 S}{\partial \xi
\partial x} \neq 0 \quad {at} \quad x=0 \ .
$$
Since $\frac{\partial^2 S}{\partial \xi \partial x} \neq 0$, we
can solve for $x$ as a function of $\xi$ on $C$ and let
$x=x(\xi)$. Since
$$
\frac{\partial^2 S}{\partial \xi^2 } \left(x(\xi), \xi\right) +
\frac{\partial^2 S}{\partial \xi^2 }\left(x(\xi),
\xi\right)x'(\xi) = 0 \ ,
$$
on $C$, we conclude that $x'(0)=0$. Now since
$$
\frac{\partial^3 S}{\partial \xi^3} \neq 0 \ ,
$$
we conclude that $x''(\xi) \neq 0$, so by a change of coordinates
on $R$ we can assume $x= \xi^2$ on $C$. Let $C^+$ be the part of
$C$ where $\xi >0$ , and $C^-$ be the part where $\xi <0$ . By the
last of  the equations (2.20) we have $\xi=+\sqrt{\rho}$ on $C^-$.
So on $C^+$ we have
$$
-\frac23 {\rho ^{\frac32}} + v_0 = S(\xi)
$$
and on $C^-$ we have
$$
\frac23 {\rho ^{\frac32}} + v_0 = S(-\xi)  \ .
$$
Since $\rho$ and $v_0$ are functions of $x$ alone, we must have
\begin{equation}
v_0(x)= \frac12 \biggl(S(\xi) + S(-\xi) \biggr), \quad \biggl(\rho(x)\biggr)^{3}
= \frac{9}{16}\biggl(S(\xi) - S(-\xi) \biggr)^2 \ ,
\end{equation}
with $x= \xi^2$ . The expressions on the right are both even
functions of $\xi$, so $v_0$ and $\rho^3$ exist by the above lemma .
To show that the cubic roots of $\rho^3$ exists we note that since
$S'(\xi)=S''(\xi)=0$, and $S'''(\xi) \neq 0$, the Taylor series
for $\Bigl(S(\xi) - S(-\xi) \Bigr)^2 $ starts with a non-zero term
of order six. Thus $\rho$ exists and is of order two with respect
to $\xi$, and of order one with respect to $x$. In particular,
$\zeta= +\sqrt{\rho}$ exists on $C$ and $\partial\zeta /\partial
\xi \neq 0$.

Now suppose $dim M > 1$. Choose coordinates $(x_1, ..., x_n)$ on
$M$, such that
$$
\frac{\partial S}{\partial \xi \partial x_1} \neq 0 \ .
$$
For $\alpha=(\alpha_2, ..., \alpha_n)$ , let $C_\alpha$ be the
intersection of $C$ with the line $x_2 = \alpha_2 , ..., x_n =
\alpha_n$. Applying the preceding argument to $C_ \alpha$, we find
functions ${v_0}^{\alpha}$, ${\rho}^{\alpha}$ and $\zeta ^ \alpha$
on $C^\alpha$ satisfying (2.20) and depending smoothly on
$\alpha$. We let $v_0$, $\rho_{1}$ and $\zeta$ be the
corresponding functions on $C$. Finally, we extend $\zeta$ from
$C$ to $M \times R$ arbitrarily. This concludes the proof of the
lemma. 
\qed
\end{proof}

To prove the Proposition $(\ref{reprfold})$, let $\psi(\textbf{x},
\xi)=v_0(\textbf{x})+\rho(\textbf{x})\zeta(\theta)- {\zeta ^3}/3$.
From (\ref{rhovfunc}) it follows easily that the critical set of
$\psi$ equals the critical set of $S$. Making the change of
coordinates $\textbf{x} \rightarrow \textbf{x}$ and $\xi
\rightarrow \zeta(\theta, \textbf{x})$, we get the phase function
of the desired form. The representation $(\ref{axi})$ follows
directly from the Malgrange preparation theorem (see, e.g.,
\cite{Ho}, vol. 1, Sec 7.5).

%%%%%%%%%%%%%%%%%%%%%%%%%%%%%%%%%%%%%%%%%%%%%%%%%%%%%%%%%%%%%%%%%%%%%%%%%%%%%%%%%%%%%%%%%%%%
%%%%%%%%%%%%%%%%%%%%%%%%%%%%%%%%%%%%%%%%%%%%%%%%%%%%%%%%%%%%%%%%%%%%%%%%%%%%%%%%%%%%%%%%%%%%%%
\chapter{Uniform stationary phase asymptotics}

We consider the integral
$$I(\lambda ,\alpha)=\int_{-\infty}^{\infty} e^{i\lambda \phi(x,\alpha)}f(x)dx, $$
where $\alpha>0,$ $\lambda$ is a large positive parameter. With
smooth f, for the case when the phase function $\phi\in
C^{\infty}$ has two stationary points, $x_1(\alpha)$ and
$x_2(\alpha),$ which approach the same limit $x_0$ when
$\alpha\rightarrow 0.$ Let $\phi_{xx}(x_1,\alpha)<0$ and
$\phi_{xx}(x_2,\alpha)>0.$

The standard stationary-phase approximation of $I(\lambda
,\alpha)$ fails:
\begin{eqnarray*}
I(\lambda ,\alpha)\approx \Bigl(\frac{2\pi}{\lambda
}\Bigr)^{1/2}\sum_{l=1,2}\frac{f(x_l(\alpha))e^{i\lambda
\phi(x_{l}( \alpha ),\alpha)}}{\mid \phi_{xx}(x_l(\alpha))\mid
^{1/2}} e^{i\frac{\pi}{4}\delta_l}
\end{eqnarray*}
with $\delta_l=sgn\phi_{xx}(x_l(\alpha)).$ In our consider
$\delta_1=-1,\delta_2=1,$ so we have
\begin{eqnarray}\label{ap1}
I(\lambda ,\alpha)=\Bigl(\frac{2\pi}{ \lambda
}\Bigr)^{1/2}\Bigl[\frac{f(x_2(\alpha))e^{(i\lambda
\phi(x_2(\alpha),\alpha)+i\pi/4)}}{\sqrt{\phi_{xx}(x_2(\alpha))}}+
\frac{f(x_1(\alpha))e^{(i\lambda
\phi(x_1(\alpha),\alpha)-i\pi/4)}}{\sqrt{\mid
\phi_{xx}(x_1(\alpha))\mid }}\Bigr]+O(\lambda ^{-1}),
\end{eqnarray}
$\lambda \rightarrow \infty.$

Assume also that $\phi(x,\alpha)$ is analytic for small $(x-x_0)$
and small $\alpha>0,$ we have
\begin{eqnarray}\label{ap20}
\phi_{xxx}\not=0,\ \phi'_{x}=\phi_{xx}=0,\
\phi_{x\alpha}\not=0
\end{eqnarray}
at $x=x_{0},\alpha=0.$

Under these conditions a theorem by Chester, Friedman and Ursell
(``An extension of the method of steepest descent", Proc. Cambr.
Phil. 1957, 53, 599-611) imples that there exists a change of
variable $x=x(\tau),$ analytic and invertible for small $(x-x_0)$
and small $\alpha>0,$ depending parametrically on $\alpha,$ such
that

\begin{eqnarray}\label{ap14}
\phi(x,\alpha)=\phi_0(\alpha)+\frac{\tau^3}{3}-\xi(\alpha)\ \tau
\end{eqnarray}
where $\phi_0(\alpha)$ and $\xi(\alpha)$ are analytic functions of
$\alpha.$

Then, we have
\begin{eqnarray}\label{ap2}
I(\lambda ,\alpha)=e^{i\lambda
\phi_0}\int_{-\infty}^{\infty}e^{i\lambda
(\tau^3/3-\tau\xi(\alpha))}f(x(\tau))\frac{dx(\tau)}{d\tau}\ d\tau
\end{eqnarray}
By a version of Malgranges preparation theorem, we have the
representation
\begin{eqnarray}\label{ap3}
f(x(\tau))\frac{dx(\tau)}{d\tau}\
d\tau=A_{0}(\alpha)+B_{0}(\alpha)\tau+h(\tau)(\tau^2-\xi)
\end{eqnarray}
where $h(\tau)$ is smooth function.

Substituting $\mathrm{(\ref{ap3})}$ into $\mathrm{(\ref{ap2})}$ we
have
\begin{eqnarray*}
I(\lambda ,\alpha)= e^{i\lambda \phi_0(\alpha)}\Bigl[2\pi
A_{0}(\alpha)\lambda ^{-1/3}Ai(-\lambda ^{2/3}\xi)-2\pi i
B_{0}(\alpha)\lambda ^{-2/3} Ai'(-\lambda ^{2/3}\xi)+C(\lambda
,\xi)\Bigr]
\end{eqnarray*}
where $$C(\lambda
,\xi)=\frac{i}{\lambda}\int_{-\infty}^{\infty}h'(\tau)e^{i\lambda
(\tau^3/3-\tau\xi)}d\tau=O(\lambda ^{-4/3})$$ Integrating the
integral for C as many time as we wish we obtain
\begin{eqnarray}\label{ap4}
I(\lambda ,\alpha)=e^{i\lambda \phi_0(\alpha)}\Bigl[2\pi A \lambda
^{-1/3}Ai(-\lambda ^{2/3}\xi) -2\pi i B\lambda ^{-2/3}Ai'(-\lambda
^{2/3}\xi)\Bigr]
\end{eqnarray}
where
\begin{eqnarray*}
A=\sum_{n=0}^{\infty}A_n(\alpha)\Bigl(\frac{i}{\lambda }\Bigr)^n,\
B=\sum_{n=0}^{\infty}B_n(\alpha)\Bigl(\frac{i}{\lambda }\Bigr)^n
\end{eqnarray*}

In order to compute $\phi_{0}(\alpha),\xi(\alpha)$ and the leading
coefficients $A_{0}(\alpha),B_{0}(\alpha)$ we use the principle of
asymptotic matching.

We fix $\alpha>0$ and we consider $\lambda \rightarrow \infty.$
Then the asymptotics of $Ai,Ai'$ read as follows
\begin{eqnarray}\label{ap5}
Ai(-\lambda ^{2/3}\xi)\approx \frac{1}{2\sqrt{\pi}}\lambda
^{-1/6}\xi^{-1/4}\Bigl[e^{2i\lambda \xi^{3/2}/3-i \pi/4}+
e^{-2i\lambda \xi^{3/2}/3+i \pi/4 }\Bigr]
\end{eqnarray}

\begin{eqnarray}\label{ap6}
Ai'(-\lambda ^{2/3}\xi)\approx \frac{-1}{2\sqrt{\pi}}\lambda
^{1/6}\xi^{1/4}\Bigl[e^{2i\lambda \xi^{3/2}/3+i \pi/4}+
e^{2i\lambda \xi^{3/2}/3-i \pi/4 }\Bigr]
\end{eqnarray}

Substituting $\mathrm{(\ref{ap5})}$ and $\mathrm{(\ref{ap6})}$
into $\mathrm{(\ref{ap4})},$  we get the expression
\begin{eqnarray}\label{ap7}
I(\lambda ,\alpha)&\approx &\Bigl(\frac{\pi}{i\lambda
}\Bigr)^{1/2}(A_{0}\xi^{-1/4}-B_{0}\xi^{1/4})\ e^{i\lambda
(\phi_{0}+2\xi^{3/2}/3)}\nonumber\\
&=&\Bigl(\frac{\pi}{i\lambda
}\Bigr)^{1/2}(A_{0}\xi^{-1/4}+B_{0}\xi^{1/4})\ e^{i\lambda
(\phi_{0}-2\xi^{3/2}/3)}+O(\lambda^ {-3/2})
\end{eqnarray}

The principle of asymptotic matching requires that the expansion
$\mathrm{(\ref{ap7})}$ must coincide with the non-uniform
expansion $\mathrm{(\ref{ap1})}.$ Comparing these expressions, and
taking into account that $\phi(x_1,\alpha)>\phi(x_2,\alpha)$ we
obtain
\begin{eqnarray}
\phi_0+\frac{2}{3}\ \xi^{3/2}=\phi(x_1,\alpha)
\end{eqnarray}

\begin{eqnarray}
\phi_0-\frac{2}{3}\ \xi^{3/2}=\phi(x_2,\alpha)
\end{eqnarray}
for the phases, and
\begin{eqnarray}
A_{0}\xi^{-1/4}+B_{0}\xi^{1/4}=\sqrt{2}\ \frac{f(x_2)}{
(\phi_{xx}(x_2,\alpha))^{1/2}}
\end{eqnarray}

\begin{eqnarray}
A_{0}\xi^{-1/4}-B_{0}\xi^{1/4}=\sqrt{2}\ \frac{f(x_1)}{\mid
\phi_{xx}(x_1,\alpha)\mid ^{1/2}}
\end{eqnarray}
which give
\begin{eqnarray}
\phi_0(\alpha)=\frac{1}{2}\Bigl(\phi(x_1(\alpha),\alpha)+\phi(x_2(\alpha),\alpha)\Bigr)
\end{eqnarray}

\begin{eqnarray}\label{ap13}
\xi(\alpha)=\Big[\
\frac{3}{4}(\phi(x_1(\alpha),\alpha)-\phi(x_2(\alpha),\alpha))\
\Bigr]^{2/3}
\end{eqnarray}

We want now to approximate $\xi(\alpha),\ \alpha\rightarrow 0^+$.
For this we set $x_{0}=0$ (which amounts for changing the variable
$x$ to $x'=x-x_0$) and we expand $\phi(x,\alpha)$ near
$(x=0,\alpha=0),$
\begin{eqnarray}\label{ap13}
\phi(x,\alpha)&=&\phi(0,0)+\phi_{x}(0,0)\ x+\phi_{\alpha}(0,0)\ \alpha\nonumber\\
&+&\frac{1}{2}\ \phi_{xx}(0,0)\ x^2+\phi_{x\alpha}\ \alpha
x+\frac{1}{2}\
\phi_{\alpha\alpha}\ \alpha^2\nonumber\\
&+& \frac{1}{6}\ \phi_{xxx}(0,0)\ x^3+\frac{1}{2}\
\phi_{xx\alpha}\ x^2\alpha
+\frac{1}{2}\ \phi_{x\alpha\alpha}\ x\alpha^2\nonumber\\
&+&\frac{1}{6}\ \phi_{\alpha\alpha\alpha}(0,0)\
\alpha^3+({4^{th}-order \
terms})\nonumber\\
&=&\phi+\phi_{\alpha}\ \alpha+\phi_{x\alpha}\ \alpha
x+\frac{1}{2}\
\phi_{\alpha\alpha}\ \alpha^2+\frac{1}{6}\ \phi_{xxx}\ x^3\nonumber\\
&+&\frac{1}{2}\ \phi_{xx\alpha}\ x^2\alpha+\frac{1}{2}\
\phi_{x\alpha\alpha}\ x \alpha^2+\frac{1}{6}\
\phi_{\alpha\alpha\alpha}\ \alpha^2
\end{eqnarray}
Differentiating the last equation we have
\begin{eqnarray}\label{ap8}
\phi_{x}(x,\alpha)=\phi_{x\alpha}\ \alpha+\frac{1}{2}\ \phi_{xxx}\
x^2+\phi_{xx\alpha}\ x\alpha+\frac{1}{2}\ \phi_{x\alpha\alpha}\alpha^2
\end{eqnarray}
and
\begin{eqnarray}\label{ap9}
\phi_{xx}(x,\alpha)=\phi_{xxx}\ x+\phi_{xx\alpha}\ \alpha
\end{eqnarray}

We compute $x_1(\alpha)$, $x_2(\alpha)$ by solving the equation
(approximate eq. $\mathrm{(\ref{ap8})},$  $O(\alpha^2)$)
\begin{eqnarray}
\phi_{x}(x,\alpha)=\frac{1}{2}\ \phi_{xxx}(0,0)\
x^2+\phi_{xx\alpha}(0,0)\ x\alpha+\phi_{x\alpha}(0,0)\ \alpha=0
\end{eqnarray}
The roots of the last equation are
\begin{eqnarray*}
x_{1,2}(\alpha)&=&\frac{-\alpha\ \phi_{xx\alpha}\pm
(\phi_{xx\alpha}^2\
\alpha^2-2\phi_{xxx}\ \phi_{x\alpha}\ \alpha)^{1/2}}{\phi_{xxx}}\\
&\approx &\pm (-2\phi_{xxx}\ \phi_{x\alpha}\
\alpha)^{1/2}(\phi_{xxx})^{-1}
\end{eqnarray*}
since $\sqrt{\alpha}>>\alpha>\alpha^2$ as $\alpha\rightarrow 0^+.$
The assumptions
$\phi''(x_1(\alpha),\alpha)<0,\phi''(x_2(\alpha),\alpha)>0$ and
equation $\mathrm{(\ref{ap9})}$ , for $\alpha\rightarrow 0^+,$
imply that
\begin{eqnarray}\label{ap10}
x_1(\alpha)\approx -(-2\phi_{xxx}\ \phi_{x\alpha}\
\alpha)^{1/2}(\phi_{xxx})^{-1}
\end{eqnarray}
\begin{eqnarray}\label{ap11}
x_2(\alpha)\approx +(-2\phi_{xxx}\ \phi_{x\alpha}\
\alpha)^{1/2}(\phi_{xxx})^{-1}
\end{eqnarray}

We also need to approximate the difference
\begin{eqnarray}\label{ap12}
\delta\phi=\phi(x_1(\alpha),\alpha)-\phi(x_2(\alpha),\alpha)\approx
\frac{1}{6}\ \phi_{xxx}\ ({x_1}^3-{x_2}^3)+\phi_{x\alpha}\
\alpha(x_1-x_2)+\ldots
\end{eqnarray}
From  $\mathrm{(\ref{ap10})},$  $\mathrm{(\ref{ap11})}$ and $\mathrm{(\ref{ap12})}$
we have (for $\alpha>0$)
$${x_1}^3-{x_2}^3=-2(\phi_{xxx})^{-3}(-2\phi_{xxx}\phi_{x\alpha})^{3/2} \alpha^{3/2}$$
$$x_1-x_2=-2(\phi_{xxx})^{-1}(-2\phi_{xxx}\ \phi_{x\alpha})^{1/2}\alpha^{1/2}$$
and therefore
\begin{eqnarray}
\delta\phi&=&-2\alpha^{3/2}(\phi_{xxx})^{-1}(-2\phi_{xxx}\
\phi_{x\alpha})^{1/2}\nonumber\\
&\cdot &\Bigl[\ \frac{1}{6}(\phi_{xxx})^{-1}(-2\phi_{xxx}\
\phi_{x\alpha})+\phi_{x\alpha}\ \Bigr]\nonumber\\
&=&-\frac{2}{3}\ \alpha^{3/2}(-2\phi_{xxx}\
\phi_{x\alpha})^{3/2}(\phi_{xxx})^{-2}
\end{eqnarray}

Then, using $\mathrm{(\ref{ap13})}$ we obtain
$$\xi=\Bigl(\frac{3}{4}\ \delta\ \phi\Bigr)^{2/3}\approx
2^{-2/3}\alpha(-2\phi_{xxx}\ \phi_{x\alpha})(\phi_{xxx})^{-4/3}$$
\begin{eqnarray}
\xi=-2^{1/3}\phi_{x\alpha}\ \phi_{xxx}^{-1/3}\alpha
\end{eqnarray}
($\xi\approx
-\phi_{x\alpha}\Bigl(\frac{\phi_{xxx}}{2}\Bigr)^{-1/3}\alpha$)

Going now back to $\mathrm{(\ref{ap13})}$ we have for
$\alpha\rightarrow 0^+,$
\begin{eqnarray}
\phi_{x}(x,\alpha)\approx \phi(0,0)+\phi_{x\alpha}\
x\alpha+\frac{1}{6}\phi_{xxx}\ x^3
\end{eqnarray}
and comparing with $\mathrm{(\ref{ap14})}$ we obtain that as
$\alpha\rightarrow 0^+,$
$$\phi_{0}(\alpha)\approx \phi(0,0)$$
$$\xi(\alpha)\approx -\phi_{x\alpha}\ \Bigl(\frac{\phi_{xxx}}{2}\Bigr)^{-1/3}\alpha$$
$$\tau^3\approx -\frac{1}{2}\ \phi_{xxx}\ x^3\Rightarrow \tau\approx
-\Bigl(\frac{\phi_{xxx}}{2}\Bigr)^{1/3}x$$

Also using $\phi_{xx}(x,\alpha)=\phi_{xxx}(0,0)\
x+\phi_{xx\alpha}(0,0)\ \alpha,$ we have
\begin{eqnarray}
\phi_{xx}(x_1(\alpha),\alpha)&=&-\phi_{xxx}\ (-2\phi_{xxx}\
\phi_{x\alpha}\
\alpha)^{1/2}(\phi_{xxx})^{-1}+\phi_{xx\alpha}\ \alpha \nonumber\\
&\approx &-(-2\phi_{xxx}\ \phi_{x\alpha}\ \alpha)^{1/2}
\end{eqnarray}

\begin{eqnarray*}
\mid \phi_{xx}(x_1(\alpha),\alpha)\mid \ \approx (-2\phi_{xxx}\
\phi_{x\alpha}\ \alpha)^{1/2}
\end{eqnarray*}

\begin{eqnarray*}
\mid \phi_{xx}(x_2(\alpha),\alpha)\mid \ \approx (-2\phi_{xxx}\
\phi_{x\alpha}\ \alpha)^{1/2}
\end{eqnarray*}
and
$$\xi^{1/4}\approx \Bigl(-\phi_{x\alpha}\
\Bigl(\frac{\phi_{xxx}}{2}\Bigr)^{-1/3}\alpha\Bigr)^{1/4}$$

%%%%%%%%%%%%%%%%%%%%%%%%%%%%%%%%%%%%%%%%%%%%%%%%%%%%%%%%%%%%%%%%%%%%%%%%%%%%%%%%%%%%%%%%%%%%%%%%%%%%%%%%%%%%%%%%%%%%%%%%%%%%%%%%%%%%%%%%%%%%%%%%
%%%%%%%%%%%%%%%%%%%%%%%%%%%%%%%%%%%%%%%%%%%%%%%%%%%%%%%%%%%%%%%%%%%%%%%%%%%%%%%%%%%%%%%%%%%%%%%%%%%%%%%%%%%%%%%%%%%%%%%%%%%%%%%%%
\chapter{The fundamental solution of the Airy
equation}

We consider the semiclassical Airy equation with a point source
\begin{equation}\label{airyd}
\e^2 {{u^{\e}}^{\prime\prime}(x)}+xu^{\e}(x)= \sigma
\delta(x-x_0),\quad x\in R \ ,
\end{equation}
where the constant $\sigma$ depends on $\e$, and it will be
defined for $u^{\e}(x)$ to have appropriate asymptotics at
$x=+\infty$.

The solution of $(\ref{airyd})$  is given by
\begin{eqnarray}
&u^{\e}(x)=i\sigma \pi\Biggl( Ai\left(-\e^{-2/3}x_0 \right)-i
Bi\left(-\e^{-2/3}x_0\right)\Biggr)\e^{-4/3}Ai\left(-\e^{-2/3}x
\right) \ ,   \  x<x_0 \\
&u^{\e}(x)=i\sigma \pi Ai\left(-\e^{-2/3}x_0 \right)\Biggl(
Ai\left(-\e^{-2/3}x \right)-i Bi\left(-\e^{-2/3}x
\right)\Biggr)\e^{-4/3} \ , \   x>x_0 \ ,
\end{eqnarray}
where $Ai(z) \ , \ Bi(z)$ are the Airy functions of the first and
second kind which are the two linearly independent solutions os
the homogeneous Airy equation $w''(z)-zw(z)=0$ (see,
e.g. \cite{Leb}).

We can now choose the constant $\sigma$ so that $u^{\e}(x) =
O_{\e} (1)$, as $\e \rightarrow 0 $, for $x \rightarrow + \infty
$. This choice leads to the value $\sigma = -ie^{-i \pi/4} \e$.

Moreover, if we consider the solution in the region $x < x_0$ and
we approximate the coefficient of the Airy function
$Ai\left(-\e^{-2/3}x \right)$, that is
$$
i\sigma \pi\Biggl( Ai\left(-\e^{-2/3}x_0 \right)-i
Bi\left(-\e^{-2/3}x_0\right)\Biggr)\e^{-4/3} \ ,
$$
for small $\e$ we get the approximation
\begin{equation}\label{innerairy}
u^{\e}(x)\approx \pi^{1/2} e^{-i \pi/2} \Bigl(x_0^{-1/4}e^{i
\frac{1}{\e} \frac{2}{3}x_{0}^{3/2}} \Bigr) \e^{-1/6}
Ai\left(-\e^{-2/3}x \right) \ , \nonumber
\end{equation}
Furthermore, invoking the asymptotics of the Airy function for
small $\e$ and fixed $x$ in $(\ref{innerairy})$, we arrive to the
expansion
\begin{eqnarray} \label{wkbfundairy}
\ue_{WKB}(x)&\approx& \frac12 e^{-i \pi/4}\Bigl( x_0^{-1/4}e^{i
\frac{1}{\e} \frac{2}{3}x_{0}^{3/2}}\Bigr) \Bigl( -i x^{-1/4}e^{i
\frac{1}{\e} \frac{2}{3}x^{3/2}} + x^{-1/4}e^{-i
\frac{1}{\e} \frac{2}{3}x^{3/2}} \Bigr) \nonumber\\
&=&\Biggl(\frac12 e^{-i \pi /4} x_0^{-1/2}\Biggr) \Biggl( -i
\frac{x_0^{1/4}}{ x^{1/4}} \  e^{i \frac{1}{\e}
\left(\frac{2}{3}x_{0}^{3/2} +\frac{2}{3}x^{3/2}\right)} +
\frac{x_0^{1/4}}{ x^{1/4}} \  e^{i \frac{1}{\e}
\left(\frac{2}{3}x_{0}^{3/2} -\frac{2}{3}x^{3/2}\right)}\Biggr)
\ .
\end{eqnarray}

Therefore, we have constructed the fundamental solution of the
semiclassical Airy equation in $0<x<x_0$,
$$
u^{\e}(x)=i\sigma \pi\Biggl( Ai\left(-\e^{-2/3}x_0 \right)-i
Bi\left(-\e^{-2/3}x_0\right)\Biggr)\e^{-4/3}Ai\left(-\e^{-2/3}x
\right) \ ,
$$
which has the desired compound asymptotics at infinity, namely
$$
u^{\e}(x)= O(1) \ , \ \ x \rightarrow +\infty \ , \ \ \e
\rightarrow 0 \ ,
$$
and it leads to a WKB approximation  which remains $O(1)$, as $\e
\rightarrow 0$ in the region $0<x<x_0$.

We close this Appendix by noting that for large $x_0 \gg \e$, we
can approximate the fundamental solution for any $x>0$ by
$(\ref{innerairy})$, a fact which is used in checking the
uniformization process of the WKB approximation to the Airy
equation near the caustics.

\renewcommand{\bibname}{References}
%\bibliography{mybib}
%\begin{thebibliography}{plain}

\end{document}